\documentclass[10pt,journal,compsoc]{IEEEtran}
\usepackage{amsfonts}
\usepackage{amsmath}
\usepackage{cases}
\usepackage{mathrsfs}
\usepackage{bbm}
\usepackage{amsfonts}
\usepackage{multirow}
\usepackage{caption}
\usepackage{color}
\usepackage{flushend}
\usepackage{float}
\usepackage{hyperref}
\usepackage[ruled,vlined,linesnumbered]{algorithm2e}
\usepackage{algorithmic}
\usepackage{subfigure}
\usepackage{extarrows}
\usepackage{graphicx,graphics,txfonts}
\usepackage{array}
\usepackage{tabularx}
\usepackage{stmaryrd}
\begin{document}

\newtheorem{definition}{Definition}
\renewcommand{\algorithmicrequire}{\textbf{Requires}}
\newtheorem{theorem}{Theorem}
\newtheorem{lemma}{Lemma}
\newtheorem{axiom}{Axiom}
\newtheorem{example}{Example}
\newtheorem{corollary}{Corollary}
\newtheorem{property}{Property}
\newtheorem{proof}{Proof}

\newcommand{\partitle}[1]{\medskip \noindent \textbf{#1.}}
\newcommand{\subpartitle}[1]{\medskip \emph{#1.}}
\newcommand{\topcaption}{%
\setlength{\abovecaptionskip}{0pt}%
\setlength{\belowcaptionskip}{0pt}%
\caption}

\title{Skyline Diagram: Efficient Space Partitioning for Skyline Queries}

\author{Jinfei Liu,
        Juncheng Yang,
        Li Xiong,
        Jian Pei,~\IEEEmembership{fellow,~IEEE,}
        Jun Luo,
        Yuzhang Guo,
        Shuaicheng Ma,
        and Chenglin Fan
\IEEEcompsocitemizethanks{
\IEEEcompsocthanksitem Jinfei Liu, Li Xiong, Yuzhang Guo, and Shuaicheng Ma are with Emory University.\protect\\
E-mail: \{jinfei.liu, lxiong, yuzhang.guo, and shuaicheng.ma\}@emory.edu}

\IEEEcompsocitemizethanks{ \IEEEcompsocthanksitem Juncheng Yang is with Carnegie Mellon University.\protect\\
E-mail: jasonyang@cmu.edu}

\IEEEcompsocitemizethanks{ \IEEEcompsocthanksitem Jian Pei is with Simon Fraser University.\protect\\
E-mail: jpei@cs.sfu.ca}

\IEEEcompsocitemizethanks{ \IEEEcompsocthanksitem Jun Luo is with Machine Intelligence Center, Lenovo Group Limited.\protect\\
E-mail: jluo1@lenovo.com}

\IEEEcompsocitemizethanks{ \IEEEcompsocthanksitem Chenglin Fan is with UT Dallas.\protect\\
E-mail: cxf160130@utdallas.edu}

\thanks{Manuscript received XXXXXX; revised XXXXXX.}
}

\IEEEtitleabstractindextext{
\begin{abstract}
Skyline queries are important in many application domains. In this paper, we propose a novel structure \emph{Skyline Diagram}, which given a set of points, partitions the plane into a set of regions, referred to as skyline polyominos. All query points in the same skyline polyomino have the same skyline query results. Similar to \emph{$k^{th}$-order Voronoi diagram} commonly used to facilitate $k$ nearest neighbor ($k$NN) queries, skyline diagram can be used to facilitate skyline queries and many other applications. However, it may be computationally expensive to build the skyline diagram. By exploiting some interesting properties of skyline, we present several efficient algorithms for building the diagram with respect to three kinds of skyline queries, quadrant, global, and dynamic skylines. In addition, we propose an approximate skyline diagram which can significantly reduce the space cost. Experimental results on both real and synthetic datasets show that our algorithms are efficient and scalable.
\end{abstract}

\begin{IEEEkeywords}
Skyline, Voronoi, Diagram, Queries.
\end{IEEEkeywords}
}

\maketitle
\section{Introduction}\label{sec:introduction}
Similarity queries are fundamental queries in many applications which retrieve similar objects given a query object. One class of the similarity queries, $k$NN queries, have been extensively studied which retrieve the $k$ nearest (or most similar) objects based on a predefined distance or similarity metric. For objects with multiple attributes, the similarity or distance on different attributes are typically aggregated with predefined weights. In many scenarios, it may not be clear how to define the relative weights in order to aggregate the attributes. {\em Skyline}, also known as {\em Maxima} in computational geometry or {\em Pareto} in business management, is important for multi-criteria decision making or multi-attribute similarity retrieval. Without assuming any relative weights of the attributes, the skyline of a set of multi-dimensional data points consists of all objects that are not dominated by any others, i.e., no other objects are better (or more similar to the query object) in at least one dimension and at least as good (as similar) in all dimensions.

\partitle{Skyline Queries}
There are many example applications that skyline queries may be desired. For instance, a physician who is treating a heart disease patient may wish to retrieve similar patients based on their demographic attributes and diagnosis test results in order to enhance and personalize the treatment for the patient. A car dealer who wishes to price a used car competitively may attempt to retrieve all similar cars (competitors) on the market based on a set of attributes such as mileage and year.  For simplicity, we use the running example below to illustrate the skyline definition as well as algorithm descriptions throughout the paper.

\begin{figure}[htb]
 \centering
 \includegraphics[width=0.3\textwidth]{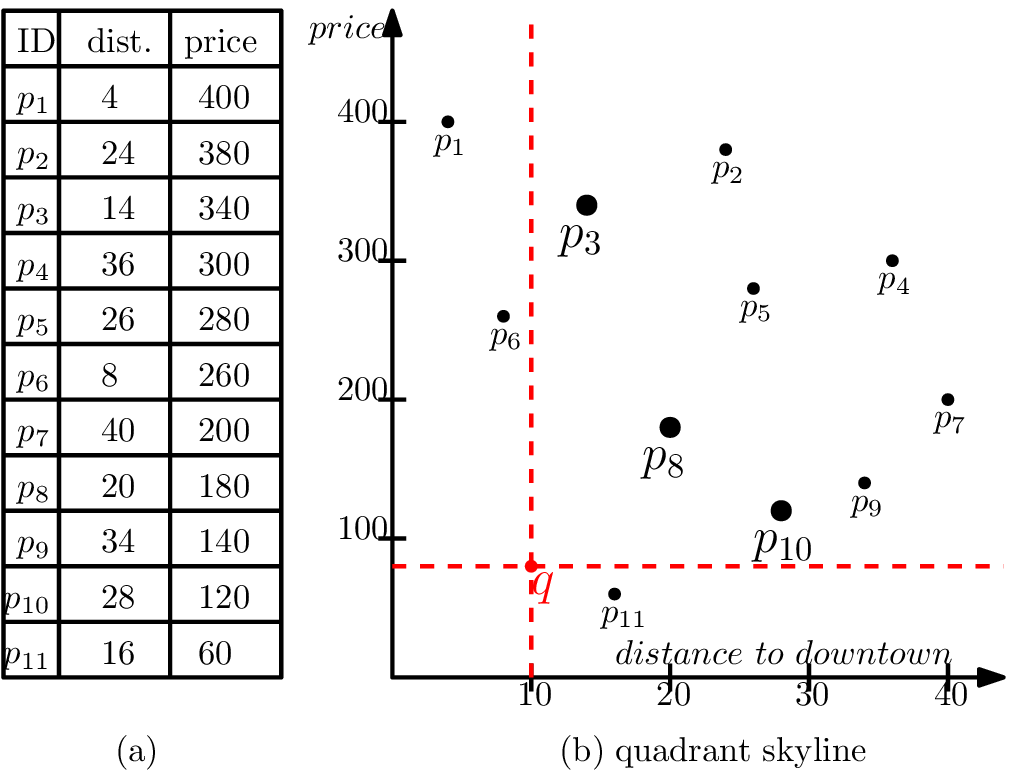}
 \includegraphics[width=0.3\textwidth]{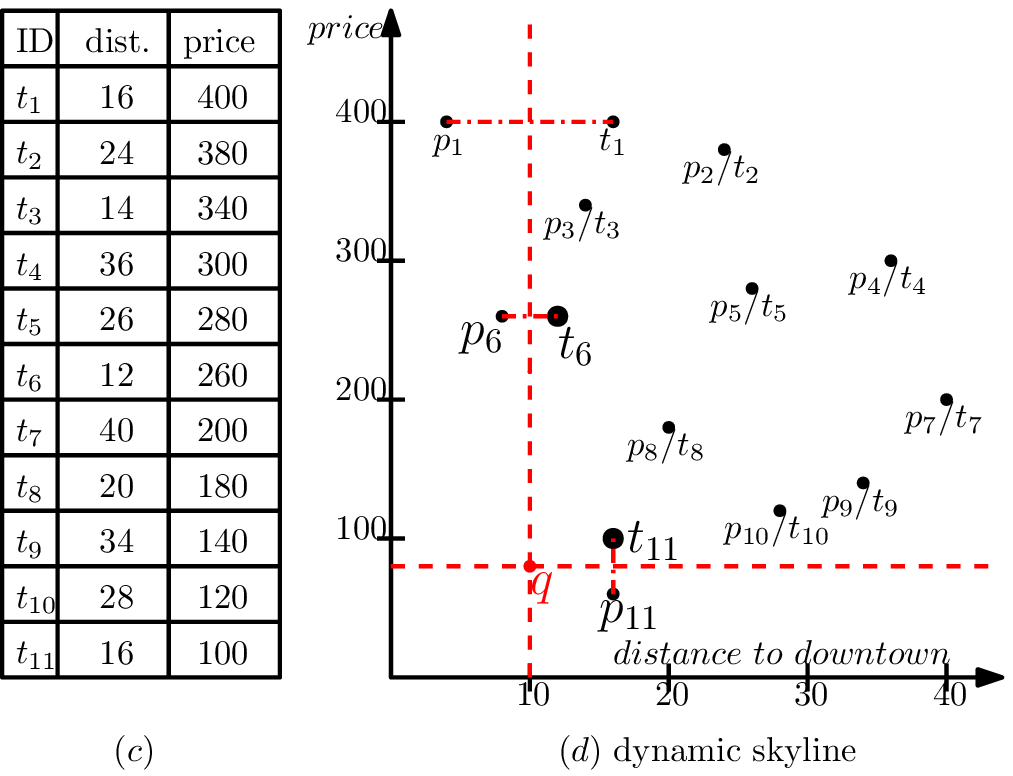}
 \vspace{-0.08in}
 \captionsetup{font={scriptsize}}\caption{A skyline example of hotels.}
 \label{fig:skyH}
\end{figure}

Consider a hotel manager who wishes to retrieve all competing hotels that are similar to his/her hotel with respect to price and distance to downtown.  Figure \ref{fig:skyH}(a) illustrates a dataset $P=\{p_1,p_2,...,p_{11}\}$, each point representing a hotel with two attributes: the distance to downtown and the price. Figure \ref{fig:skyH}(b) shows the corresponding points in the two-dimensional space.

Given a query hotel $q=(10,80)$, if we only consider the hotels with higher price and longer distance to downtown, i.e., the points in the first quadrant with $q$ as the origin, the skyline points are $p_3,p_8,p_{10}$ as shown in Figure \ref{fig:skyH}(b) (we refer to this as {\em quadrant skyline}).
If we consider all hotels, we can compute the skyline in each quadrant independently, i.e., only considering dominance within each quadrant, and take the union which is $p_3,p_8,p_{10}, p_6, p_{11}$ (we refer to this as {\em global skyline}). Alternatively, if we consider the absolute difference to the query point on each dimension, hence a point can dominate another point in a different quadrant, we have {\em dynamic skyline}\footnote{we follow the name conventions in the literature \cite{DBLP:conf/vldb/DellisS07} for these different types of skyline queries.}. To compute dynamic skyline, we can map all data points to the first quadrant with $q$ as the origin and the distance to $q$ as the mapping function, and then compute the traditional skyline from all the mapped points. The mapped points with $t_i[j]=|p_i[j]-q[j]|+q[j]$ on each dimension $j$ are shown in Figure \ref{fig:skyH}(c) and (d). It is easy to see that $t_6$ and $t_{11}$ are skyline in the mapped space, which means $p_6$ and $p_{11}$ are the dynamic skyline with respect to query $q$. We note that dynamic skyline is always a subset of global skyline since the mapped points may dominate some points that are otherwise global skyline.

\partitle{Skyline Diagram}
Given the importance of such skyline queries, it is desirable to precompute the skyline for any random query point to facilitate and expedite such queries in real time. Voronoi diagram \cite{DBLP:journals/tc/ChazelleE87} is commonly used to compute and facilitate $k$NN queries.  Inspired by the Voronoi diagram which captures the regions with same $k$NN query results, we propose a fundamental structure in this paper, referred to as {\em skyline diagram}, to capture the query regions with the same skyline result and to facilitate skyline queries.

\begin{figure}[htb]
\begin{minipage}[t]{0.48\linewidth}
\centering
\includegraphics[width=1.6in]{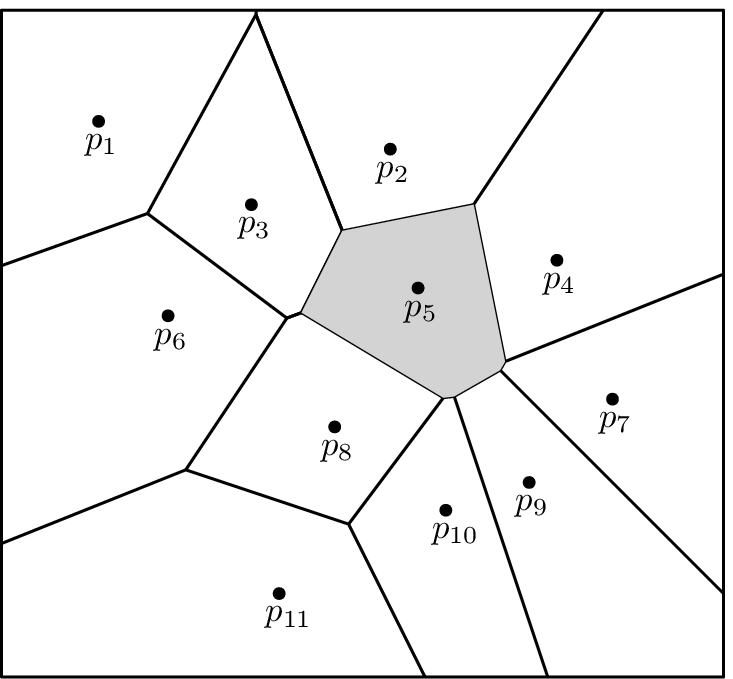}
\captionsetup{font={scriptsize}}
\caption{Voronoi diagram of $k$NN queries.} \label{fig:voronoi}
\end{minipage}%
\begin{minipage}[t]{0.48\linewidth}
\centering
\includegraphics[width=1.6in]{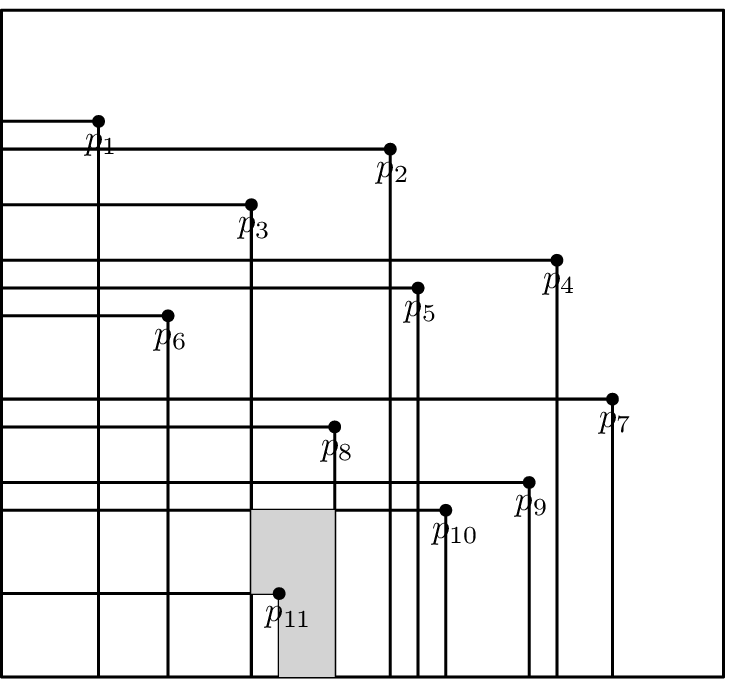}
\captionsetup{font={scriptsize}}
\caption{Skyline diagram of quadrant skyline.} \label{fig:skycell}
\end{minipage}
\end{figure}

Given a set of points (seeds), Voronoi diagram (as shown in Figure \ref{fig:voronoi}) partitions the plane into a set of polygons corresponding for each point, each \emph{query point} in the region is closer to the point than to any other points. These regions are called Voronoi cells. In other words, the query points in the same Voronoi cell have the same nearest neighbor which is the point in the cell. For example, the query points in the shaded region in Figure \ref{fig:voronoi} have $p_5$ as the nearest neighbor. This is the case of $k$NN query where $k=1$, similarly, $k^{th}$-order Voronoi diagram can be built for $k$NN queries ($k>1$), where the query points in each Voronoi cell have the same $k$NN results (may not correspond to the point in the cell as in the Voronoi diagram).

Analogously, given a set of points (seeds), our proposed skyline diagram partitions the plane into a set of regions, which we call \emph{skyline polyominos}, and the query points in each skyline polyomino have the same skyline results. Figure \ref{fig:skycell} shows an example skyline diagram for quadrant skyline queries given the same points. The query points in the shaded region have the same skyline result of $p_8,p_{10}$.

Given the precomputed skyline diagram, skyline queries can be quickly answered in real time. In addition, it can be used for other applications such as: 1) to facilitate the computation of reverse skyline queries \cite{DBLP:conf/vldb/DellisS07, DBLP:journals/tkde/WangXCL12}, similar to using Voronoi diagram for reverse $k$ nearest neighbor (R$k$NN) queries \cite{DBLP:journals/pvldb/SharifzadehS10}, 2) to authenticate skyline results from outsourced computation, similar to using Voronoi diagram for authenticating $k$NN queries \cite{DBLP:conf/icde/YiuLY11}, and 3) to enable efficient Private Information Retrieval (PIR) based skyline queries, similar to using Voronoi diagram for PIR based kNN queries \cite{DBLP:conf/dasfaa/WangMHX15}.

\partitle{Challenges}
While there are many applications of skyline diagram, it is non-trivial to compute the diagram.  For quadrant or global skyline queries, a straightforward approach is to draw vertical and horizontal grid lines crossing each point, which divides the plane into $O(n^2)$ cells. We can easily show that each of these cells has the same skyline since there are no points within the cell that would change the dominance relationship of the points.  Thus, we can compute the skyline for each cell, each requiring $O(n\log n)$ time. The time complexity of such a baseline algorithm is $O(n^3\log n)$ which is not efficient.

The time complexity of computing the skyline diagram for dynamic skyline can be significantly higher. Because of the mapping function, a straightforward approach is to draw horizontal and vertical bisector lines of each pair of points on each dimension, in addition to the grid lines crossing each point. These resulting subcells are guaranteed to have the same dynamic skyline since there are no points or mapped points in each subcell that would change the dominance relationship of the points. Since the plane is divided into $O({n\choose 2}^2)$ subcells, such a baseline algorithm requires $O(n^5\log n)$ complexity which is prohibitively high.

\partitle{Contributions} In this paper, we formally define a novel structure, skyline diagram, which enables precomputation of skyline queries as well as other applications. We study the skyline diagram with respect to three different skyline query definitions, quadrant, global, and dynamic skyline, and propose efficient algorithms. To facilitate the presentation, we focus on the algorithms for two-dimensional space first and if not specifically mentioned, all time complexities refer to the case of two dimensions, then briefly show that our proposed algorithms are extensible to high-dimensional space. We summarize our contributions as follows.

\begin{itemize}
\item We define a novel structure, skyline diagram, to enable precomputation of skyline queries. The skyline diagram consists of skyline regions, referred to as skyline polyominos, each of them corresponding to the same set of skyline result.  Similar to Voronoi diagram for $k$NN queries, skyline diagram has many applications including precomputation of skyline queries, reverse skyline queries, authentication of outsourced skyline queries, and PIR based skyline queries.

\item To compute the skyline diagram for quadrant/global skyline, we present a baseline algorithm with $O(n^3)$ time complexity and define an important notion of skyline cell. Furthermore, based on the observation of some interesting properties, we propose two improved $O(n^3)$ algorithms, which perform much better than the baseline algorithm in practice. Finally, we quantify the exact relationship between the skyline results of neighboring cells, and present an $O(n^2)$ sweeping algorithm which further improves the performance.

\item To compute the skyline diagram for dynamic skyline, we first present a baseline algorithm with $O(n^5)$ time complexity and define an important notion of skyline subcell. Furthermore, based on the observation that dynamic skyline query result is a subset of global skyline, we present an improved subset algorithm utilizing the skyline diagram of global skyline, which requires $O(n^5)$ but is better in practice. Finally, based on the relationship of the skyline results of neighboring subcells, we present a scanning algorithm which achieves $O(n^4\log n)$ time.

\item To significantly reduce the space cost, we propose the approximate skyline diagram by only requiring each skyline polyomino to have approximately the same skyline result. We present two heuristic algorithms, Bottom-Up Merging (BUM) algorithm and Top-Down Partitioning (TDP) algorithm, to efficiently compute the approximate skyline diagram with different tradeoffs.

\item We conduct comprehensive experiments on real and synthetic datasets. The experimental results show our proposed algorithms are efficient and scalable for both the exact skyline diagram and the approximate skyline diagram.%
\end{itemize}

\partitle{Organization}
The rest of the paper is organized as follows. Section \ref{sec:relatedwork} presents the related work. Section \ref{sec:pre} introduces some background knowledge and formally defines skyline diagram. The algorithms for computing the skyline diagram for quadrant/global skyline and dynamic skyline are presented in Sections \ref{sec:oneway} and \ref{sec:dynamic} respectively. We report the experimental results and findings in Section \ref{sec:experiments}. Section \ref{sec:conclusion} concludes the paper.
\section{Related Work}\label{sec:relatedwork}

The skyline computation problem was first studied in computational geometry \cite{DBLP:journals/jacm/KungLP75} which focused on worst-case time complexity. \cite{DBLP:conf/compgeom/KirkpatrickS85, DBLP:journals/ipl/LiuXX14} proposed output-sensitive algorithms achieving $O(n\log v)$ in the worst-case where $v$ is the number of skyline points which is far less than $n$ in general. Since the introduction of the skyline operator by B\"{o}rzs\"{o}nyi et al. \cite{DBLP:conf/icde/BorzsonyiKS01}, skyline has been extensively studied in the database field \cite{DBLP:conf/sigmod/ChanJTTZ06, DBLP:conf/vldb/DellisS07, DBLP:conf/cikm/YuQL0CZ17, DBLP:journals/tods/LianC10, DBLP:journals/pvldb/LiuXPLZ15, DBLP:conf/icde/LiuY0P17, DBLP:conf/cikm/LiuZXLL15, DBLP:conf/vldb/PeiJLY07, DBLP:conf/vldb/PeiJET05, DBLP:journals/tods/PeiYLJELWTYZ06, DBLP:journals/tkde/WangXCL12}.

The most related works to our skyline diagram are the ``safe zone'' for location-based skyline queries \cite{DBLP:journals/tkde/HuangLOT06,DBLP:conf/icde/LeeH09,DBLP:journals/tkde/LinXH13,DBLP:conf/edbt/CheemaLZZ13}. Huang et al. \cite{DBLP:journals/tkde/HuangLOT06} proposed the first work on continuous skyline query processing. Given a set of $n$ data points $<x_i,y_i;v_{xi},v_{yi};p_{i1},...,p_{im}> (i=1,...,n)$, where $x_i$ and $y_i$ are positional coordinates in two-dimensional space, $v_{xi}$ and $v_{yi}$ are the velocity in the $X$ and $Y$ dimensions, while $p_{ij}(j=1,...,m)$ are the $m$ static nonspatial attributes, which will not change with time. For a query point $q$ starting from $(x_q,y_q)$ moving with $(v_{qx},v_{qy})$, $q$ poses continuous skyline query while moving, and the queries involve both distance and all other static dimensions. Such queries are dynamic due to the change in spatial variables. In their solution, they compute the skyline for $x_q,y_q$ at the start time $0$. Subsequently, continuous query processing is conducted for each user by updating the skyline instead of computing from scratch. Lee et al. \cite{DBLP:conf/icde/LeeH09} studied a similar problem to \cite{DBLP:journals/tkde/HuangLOT06}. Both of them rely on the assumption that the velocities of the moving points are known. Generally speaking, they compute the skyline for query points moving on a line segment. Lin et al. \cite{DBLP:journals/tkde/LinXH13} studied a problem of computing the skyline for a range. They employed the similar idea for authenticating skyline queries in \cite{DBLP:conf/cikm/LinXH11, DBLP:journals/tkde/LinXHL14}. Cheema et al. \cite{DBLP:conf/edbt/CheemaLZZ13} proposed a safe zone for a query point $q$. A safe zone is the area such that the results of a query $q$ remain unchanged as long as the query lies inside the area. Both \cite{DBLP:journals/tkde/LinXH13} and \cite{DBLP:conf/edbt/CheemaLZZ13} studied the location-based skyline problem with $m$ static attributes and one dynamic attribute, which is the distance to the query point.

The main difference between the above work and our skyline diagram \cite{DBLP:conf/icde/LiuY0PL18} is that they only consider one dynamic attribute, while in our case all attributes can be dynamic. The skyline polyomino can be considered as a generalization of the safe zone in two or high-dimensional space. Furthermore, it is non-trivial to extend these query techniques from one dynamic attribute to two or high-dimensional case, as fundamentally these algorithms convert the problem to nearest neighbor queries for the single dynamic attribute and utilize Voronoi diagram. Compared to \cite{DBLP:conf/icde/LiuY0PL18}, in this paper, we propose an approximate skyline diagram which can significantly reduce the space cost while introducing a small amount of approximation in the skyline query result. Furthermore, we propose two heuristic algorithms, bottom-up merging algorithm and top-down partitioning algorithm, to efficiently compute the approximate skyline diagram with different tradeoffs.

\section{Preliminaries and Problem Definitions}\label{sec:pre}

In this section, we introduce our skyline diagram definition and related concepts as well as their properties which will be used in our algorithm design. For reference, a summary of notation is given in Table \ref{tab:notations}.

\vspace{-1.5em}

\begin{table}[h]\centering
\captionsetup{font={scriptsize}}\caption{The summary of notations.}\label{tab:notations}
\vspace{-1em}
{%
\begin{tabular}{|c|c|}
\hline
Notation & Definition\\
\hline
$P$   & dataset of $n$ points\\
\hline
$p_i[j]$ & the $j^{th}$ attribute of $p_i$\\
\hline
$q$   & query point\\
\hline
$n$ & number of points in $P$\\
\hline
$d$ & number of dimensions in $P$\\
\hline
$s_i$ & domain size of $i^{th}$ dimension\\
\hline
$C_{i,j}$ &  Cell with bottom left corner coordinate $(i,j)$\\
\hline
$Sky(C_{i,j})$  & the skyline of Cell $C_{i,j}$\\
\hline
$SC_{i,j}$ &  Subcell with bottom left corner coordinate $(i,j)$\\
\hline
$Sky(SC_{i,j})$  & the skyline of Subcell $C_{i,j}$\\
\hline
$Skyline(P')$ & the skyline of dataset $P'$\\
\hline
\end{tabular}}
\end{table}%

\begin{definition}(\textbf{Skyline}).
Given a dataset $P$ of $n$ points in $d$-dimensional space. Let $p$ and $p'$ be two different points in $P$, we say $p$ dominates $p'$, denoted by $p\prec p'$, if for all $i$, $p[i]\leq p'[i]$, and for at least one $i,~p[i]< p'[i]$, where $p[i]$ is the $i^{th}$ dimension of $p$ and $1\leq i\leq d$.  The skyline points are those points that are not dominated by any other point in $P$.
\end{definition}

\begin{definition}(\textbf{Dynamic Skyline Query} \cite{DBLP:conf/vldb/DellisS07}).\label{def:dynamicskyline}
Given a dataset $P$ of $n$ points and a query point $q$ in $d$-dimensional space. Let $p$ and $p'$ be two different points in $P$, we say $p$ dominates $p'$ with regard to the query point $q$, denoted by $p\prec p'$, if for all $i$, $|p[i]-q[i]|\leq |p'[i]-q[i]|$, and for at least one $i,~|p[i]-q[i]|<|p[i]-q[i]|$, where $p[i]$ is the $i^{th}$ dimension of $p$ and $1\leq i\leq d$. The skyline points are those points that are not dominated by any other point in $P$.
\end{definition}

The traditional skyline computation is a special case of dynamic skyline query where the query point is the origin. On the other hand, computing dynamic skyline given a query point $q$ is equivalent to computing the traditional skyline after transforming all points into a new space where $q$ is the origin and the absolute distances to $q$ are used as mapping functions.
Take Figure \ref{fig:skyH} as an example, given a query point $q=(10,80)$, $p_6$ dominates $p_1$ because $p_{6}$'s corresponding point $t_{6}$ in the mapped space dominates $p_1$'s corresponding point $t_{1}$. Because no other points can dominate $t_{6}$ and $t_{11}$, the result of dynamic skyline query given $q$ is $\{p_6,p_{11}\}$.

The dynamic skyline query considers the dominance among all points.  Given a query point, if we consider each quadrant divided by the query point independently, i.e., only consider dominance among points within the same quadrant, we can define global skyline query below.

\begin{definition}(\textbf{Global Skyline Query} \cite{DBLP:conf/vldb/DellisS07}).\label{def:globalskyline}
Given a dataset $P$ of $n$ points and a query point $q$ in $d$-dimensional space. The query point $q$ divides the $d$-dimensional space into $2^d$ quadrants. Let $p$ and $p'$ be two different points \emph{in the same quadrant} of $P$, we say $p$ dominates $p'$ with regard to the query point $q$, denoted by $p\prec p'$, if for all $i$, $|p[i]-q[i]|\leq |p'[i]-q[i]|$, and for at least one $i,~|p[i]-q[i]|<|p'[i]-q[i]|$, where $p[i]$ is the $i^{th}$ dimension of $p$ and $1\leq i\leq d$. The skyline points are those points that are not dominated by any other point in $P$.
\end{definition}

Given a query point, we refer to the global skyline from a single quadrant as {\bf {\em Quadrant Skyline Query}}.  In other words, the global skyline is the union of the quadrant skyline from all quadrants. Back to Figure \ref{fig:skyH}, given the query point $q$, the quadrant skyline is $\{p_3,p_8,p_{10}\}$ in the first quadrant, $\{p_6\}$ in the second quadrant,  $\emptyset$ in the third quadrant, and $\{p_{11}\}$ in the fourth quadrant. The global skyline is the entire set of $\{p_3,p_6,p_8,p_{10},p_{11}\}$.  It is easy to see that the dynamic skyline is a subset of the global skyline. This property will be used in our algorithm design for dynamic skyline diagram.

Similar to the definition of Voronoi cell and $k^{th}$-order Voronoi diagram for $k$NN query, we define the skyline polyomino and skyline diagram for skyline query as follows.
\begin{definition}(\textbf{Skyline Polyomino}).  A polyomino $SP_i$ is a skyline polyomino (hereinafter to be referred as skymino), if given any two query points $q_a$ and $q_b$ in $SP_i$, $q_a$'s skyline result $Sky(q_a)$ equals to $q_b$'s skyline result $Sky(q_b)$, while for any query point $q_c$ outside $SP_i$, the skyline result $Sky(q_c)$ of $q_c$ is not equal to $Sky(q_a)$.
\end{definition}

\begin{definition}(\textbf{Skyline Diagram}).
Given a dataset $P$ of $n$ points (seeds) $p_1,...,p_n$. We define the Skyline Diagram of $P$ as the subdivision of the plane into a set of polyominos with the property that any query points in the same polyomino have the same skyline query result.
\end{definition}

\partitle{Problem Statement}
Given $n$ points, we aim to compute the skyline diagram for quadrant/global skyline queries and dynamic skyline queries efficiently.

\section{Skyline Diagram of Quadrant and Global Skyline}\label{sec:oneway}

In this section, we present detailed algorithms for computing skyline diagram of quadrant in two-dimensional space and briefly show that they are extensible to high-dimensional space. Note that global skyline can be simply computed by taking a union of all quadrant skylines. We first show an $O(n^3)$ baseline algorithm and define an important notion of skyline cell, which will be used by all our proposed algorithms. We then present two improved algorithms based on directed skyline graph and relationship between neighboring cells. Both algorithms have $O(n^3)$ time complexity but they are much faster than the baseline in practice. Finally, we quantify the exact relationship between the skyline results of neighboring cells, and present an $O(n^2)$ sweeping algorithm which further improves the performance. For two-dimensional space, we use $x$ and $y$ to denote the two dimensions (instead of the $j^{th}$ attribute as listed in Table \ref{tab:notations}).

\subsection{Baseline Algorithm}
We first show a baseline algorithm for computing skyline diagram and introduce an important notion, skyline cell.  The key for computing skyline diagram is to find regions such that any query points in the same region have the same skyline result. Intuitively, we can find small regions that are guaranteed to have the same results and then merge them to form bigger regions.

\partitle{Skyline Cell}
If we draw one horizontal and one vertical line over each point, these $O(n)$ grid lines divide the plane into $O(n^2)$ cells. For example, in Figure \ref{fig:oneway}, the horizontal and vertical lines over each of the 11 points divide the plane into $144$ cells. It is clear that any query points inside each cell are guaranteed to have the same quadrant/global skyline because there are no points in the cell that would change the dominance relationship of the points with respect to the query point.  We name the cell as Skyline Cell.

\begin{definition}(\textbf{Skyline Cell}).
The horizontal and vertical lines over each point divide the plane into skyline cells.  Any query points in the same skyline cell have the same skyline results for quadrant/global skyline.
\end{definition}

\begin{figure}[H]
 \centering
 \includegraphics[width=0.3\textwidth]{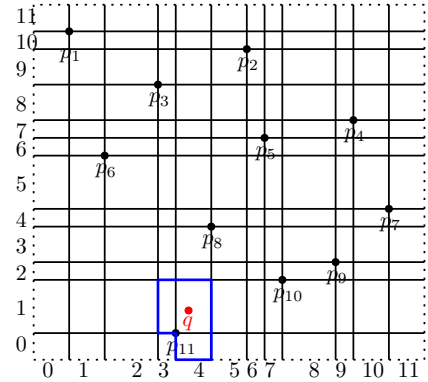}
 \captionsetup{font={scriptsize}}\caption{Quadrant skyline query.}
 \label{fig:oneway}
\end{figure}

\partitle{Finding skyline for each skyline cell}
Since we know that query points in each skyline cell have the same skyline results, we can employ any skyline algorithm to compute the skyline for each cell. Given a cell $C_{i,j}$, we denote $Sky(C_{i,j})$ as its skyline result. We can then merge the skyline cells with the same results to form skyline polyominos.  Since the skyline computation given $n$ points for each cell takes $O(n\log n)$ time and there are $O(n^2)$ skyline cells, the total time complexity is $O(n^3\log n)$. If the $n$ points are sorted on $x$-coordinate, we can compute the skyline for one cell in $O(n)$ time. Therefore, the total time can be reduced to $O(n^3)$.  This baseline algorithm is shown in Algorithm \ref{Alg:oneBasic}. After the points are sorted (Line 1), the steps for computing skyline in $O(n)$ based on ordered points are shown in Lines 5-12, where $g_{i,j}$ is the left lower intersection of skyline cell $C_{i,j}$.

\begin{algorithm}[thb]\scriptsize
\caption{The baseline algorithm for skyline diagram of quadrant skyline queries.} \label{Alg:oneBasic}
\SetKwInOut{Input}{input}\SetKwInOut{Output}{output}

\Input{a set of $n$ points and skyline cells $C_{i,j}$.}
\Output{skyline of each skyline cell $Sky(C_{i,j})$.}

sort the points in ascending order on $x$-coordinate\;
\For{i=0 to n}{
    \For{j=0 to n}{
        \For{k=1 to n}{
            \If{$p_k[x]>g_{i,j}[x]\&\&p_k[y]>g_{i,j}[y]$}{
            add $p_k$ to the candidate list\;
            }
        }
        choose the first element $p_{first}$ as the first skyline\;
        $p_{temp}=p_{first}$\;
        \For{l=2 to $|candidate~ list|$}{
            \If{$p_l[y]<p_{temp}[y]$}{
            add $p_l$ to skyline pool\;
            $p_{temp}[y]=p_l[y]$\;
            }
        }
        \Return skyline pool as $Sky(C_{i,j})$\;
}}
\end{algorithm}

\partitle{Merging skyline cells into skyline polyominos}
Once we have the skyline results for each cell, we can merge the cells with same results to form skyline polyominos.  For each skyline cell, we search its upper and right cells and combine those cells if they share the same skyline. The entire merging requires $O(n^2)$ time.

\begin{example}
In Figure \ref{fig:oneway}, the skyline cells $C_{4,0}$, $C_{4,1}$, and $C_{3,1}$ share the same skyline result $\{p_8,p_{10}\}$, and hence are combined to form a skyline polyomino.
\end{example}

\partitle{Complexity}
As we analyzed above, finding skyline phase requires $O(n^3)$ time, and merging phase requires $O(n^2)$ time. Therefore, the total time complexity for the baseline algorithm is $O(n^3)$. We have $O(n^2)$ skyline cells or skyminos and each skymino requires $O(n)$ to store. Thus, the space complexity is $O(n^3)$. The above analysis assumes attribute domain is unlimited.  In practice, the data attributes often have a domain with limited size (or can be discretized), hence the actual complexity is also bounded by the domain size (the number of possible values) of each dimension. Given a domain size $s$, the number of skyline cells is bounded by $O(min(s^2,n^2))$, hence both the time and space complexity for the baseline algorithm is $O(min(s^2,n^2)n)$. We note that the remaining algorithms have the same space complexity due to the output structure in this section.

\subsection{Directed Skyline Graph Algorithm}

In the baseline algorithm, we need to compute skyline for each skyline cell from scratch which is costly. In this subsection, based on the observation of some interesting relationships of the skyline results of neighboring cells, we propose an incremental algorithm utilizing the directed skyline
graph for computing skyline for neighboring cells. Note that the merging step of the skyline cells remains the same as the baseline.

Our algorithm is based on the key observation that when moving from one cell to its neighboring upper or right cell, the only point that will cause the skyline result to change is the point on the crossed grid line. For example, in Figure \ref{fig:oneway}, given cell $C_{0,0}$, the skyline is $\{p_1,p_6,p_{11}\}$. When moving to its right cell $C_{1,0}$ across the $p_1$ grid line, the new result is the skyline of the remaining points after removing $p_1$, that is $\{p_6,p_{11}\}$. Similarly, when moving from $C_{0,0}$ to its upper cell $C_{0,1}$ across the $p_{11}$ grid line, the new result is the skyline of the points after removing $p_{11}$, that is $\{p_1, p_6,p_{10}\}$.
Based on this observation, we propose to use a data structure called the directed skyline graph to facilitate the incremental computation of the skyline from one cell to its neighboring cell.

\begin{figure}[!htb]
\begin{minipage}[t]{0.5\linewidth}
\centering

\includegraphics[width=1.6in]{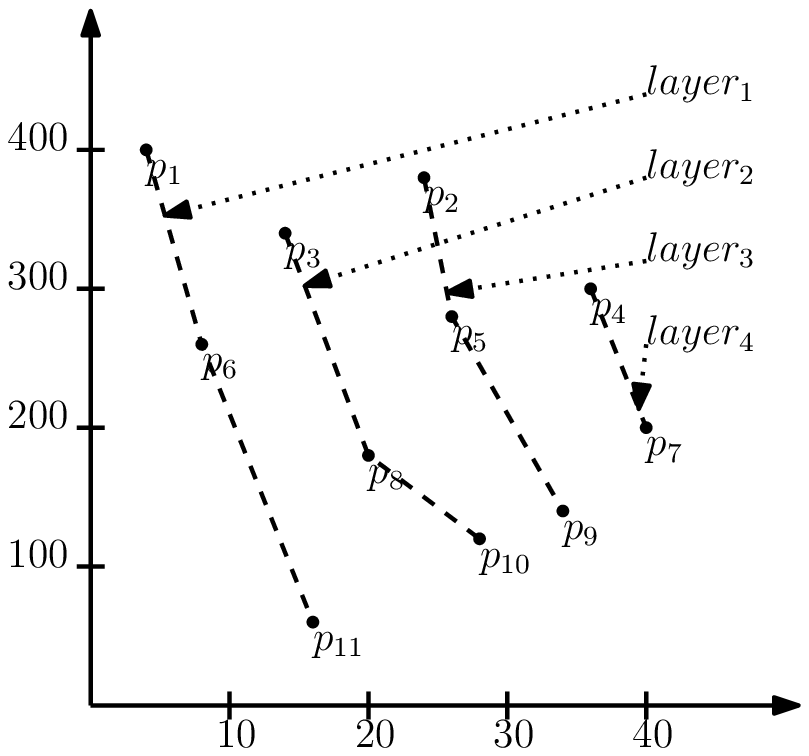}
\captionsetup{font={scriptsize}}
\caption{Skyline layers.} \label{fig:layer}
\end{minipage}%
\begin{minipage}[t]{0.5\linewidth}
\centering
\includegraphics[width=1.6in]{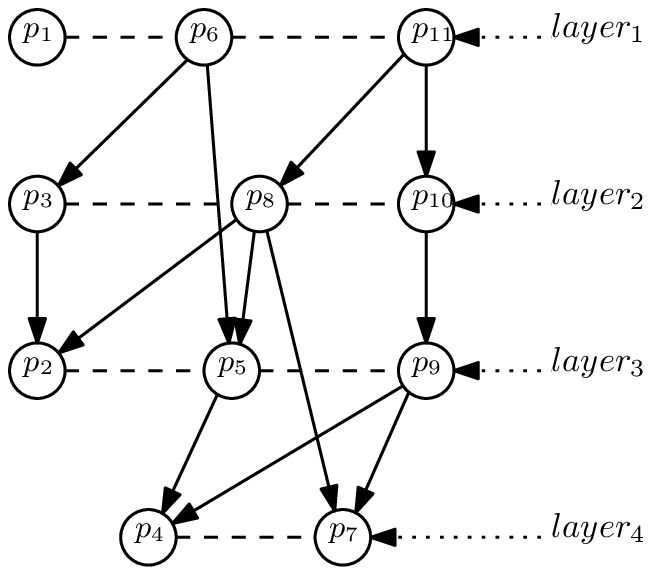}
\captionsetup{font={scriptsize}}
\caption{Directed skyline graph.} \label{fig:directedgraph}
\end{minipage}
\end{figure}

We first briefly describe the directed skyline graph (DSG) adapted from \cite{DBLP:journals/pvldb/LiuXPLZ15} and explain how it can be used to facilitate the incremental skyline computation and then present our algorithm utilizing the graph for computing the skyline for all skyline cells.

Given $n$ points, we first compute its skyline layers by employing the skyline layer algorithm from \cite{DBLP:journals/pvldb/LiuXPLZ15}. The skyline layers of our running example are shown in Figure \ref{fig:layer}. The first skyline layer consists of all skyline points in the original dataset.  The second skyline layer consists of all skyline points of the remaining points after removing the points from the first skyline layer.  And similarly for the remaining skyline layers. There are several properties for skyline layers: 1) the points on the same layer cannot dominate each other, 2) the points on a lower layer may dominate the points on a higher layer, and 3) the points on a higher layer cannot dominate the points on a lower layer. Based on these skyline layers, we obtain the directed skyline graph which captures all the direct dominance relationships between the points as shown in Figure \ref{fig:directedgraph}. For example, $p_6$ directly dominates $p_3$ and $p_5$.  We note that the directed skyline graph algorithm from \cite{DBLP:journals/pvldb/LiuXPLZ15} includes both direct and indirect dominance relationships (e.g., $p_6$ dominates $p_4$ indirectly). We adapted it such that we only include the direct links which are needed to solve our problem.

\begin{algorithm}[thb] \scriptsize \caption{The directed skyline graph algorithm for skyline diagram of quadrant skyline queries.} \label{Alg:oneDirGraph}
\SetKwInOut{Input}{input}\SetKwInOut{Output}{output}

\Input{a set of $n$ points and skyline cells $C_{i,j}$.}
\Output{skyline of each skyline cell $Sky(C_{i,j})$.}

compute the directed skyline graph DSG\;
$Sky(C_{0,0})=Sky(P)$\;

\For{i=0 to n-1}{
    tempDSG=DSG\;
    \For{j=1 to n}{
        delete the point $p_j$ between $C_{i,j-1}$ and $C_{i,j}$ from DSG\;
        delete the link between $p_j$ and its directed children\;
        $Sky(C_{i,j})$ = $Sky(C_{i,j-1})$ - $p_j$ + the children of $p_j$ without any remaining parent\;
    }
    DSG=tempDSG\;
    delete the point $p_j$ between $C_{i,0}$ and $C_{i+1,0}$ from DSG\;
    delete the link between $p_j$ and its directed children\;
    $Sky(C_{i+1,0})$ = $Sky(C_{i,0})$ - $p_j$ + the children of $p_j$ without any remaining parent\;
}
\end{algorithm}

We now show how we can incrementally compute the skyline from one cell to its neighboring cell utilizing the skyline graph. When moving from one cell to its right neighboring cell across the grid line over $p$, there are two changes in the skyline result caused by the point $p$: 1) $p$ is no longer a skyline point, 2) new skyline points may appear since they are not dominated by $p$ anymore with respect to the query point in the new cell. So all we need to do is to remove $p$ as well as its dominance links from the skyline graph, any of the children points of $p$ without remaining parents will be a new skyline (since it is no longer dominated by any points).

Given any cell, we can also compute its upper neighboring cell in a similar way. Hence our algorithm starts from the origin cell $C_{0,0}$, and incrementally computes the first row of cells from left to right. Then it incrementally computes all the rows from bottom to up. The algorithm is shown in Algorithm 2.  The directed skyline graph is computed in Line 1 and the skyline for $C_{0,0}$ is computed in Line 2. The skyline for the each row is computed in Lines 5-8. Lines 9-11 copy and update the DSG for next row.

\begin{example} Given $C_{0,0}$ in Figure \ref{fig:oneway}, its skyline is the set of points on the first skyline layer, $\{p_1,p_6,p_{11}\}$. When moving from $C_{0,0}$ to its right neighboring cell $C_{1,0}$ across the $p_1$ grid line, to compute the new skyline, all we need to do is to remove $p_1$ ($p_1$ does not have any direct dominance links), hence the skyline for $C_{1,0}$ is simply $\{p_6, p_{11}\}$ after removing $p_1$ from the skyline set.  When we move further to $C_{1,0}$'s right neighboring cell $C_{2,0}$ across the $p_6$ grid line, we just need to remove $p_6$ and remove the dominance links from $p_6$ to $p_3$ and $p_5$.  Since $p_3$ is no longer dominated by any points after $p_6$ is removed, it becomes a new skyline.  Hence the skyline for $C_{2,0}$ consists of the remaining skyline $p_{11}$ and the new skyline $p_3$, i.e., $\{p_3, p_{11}\}$.
\end{example}

\partitle{Complexity}
As we iterate through all the cells in one row, we are removing dominance links from the skyline graph.  Each link costs one update and the total number of links is $O(n^2)$.  Therefore, it requires $O(n^2)$ time to compute the skyline for cells in one row. Since there are $n$ rows, the time complexity for the directed skyline graph algorithm is $O(n^3)$.  We note that in practice, the number of links is much smaller than $n^2$. Hence the algorithm is much faster than the baseline algorithm in practice. Similar to the analysis in baseline algorithm, given a limited domain size $s$, the total number of links is $O((min\{s^2,n\})^2)$. Therefore, the time complexity for the directed skyline graph algorithm is $O((min\{s^2,n\})^2n)$. The space complexity stays the same as the baseline algorithm which is $O(min(s^2,n^2)n)$.

\subsection{Scanning Algorithm}

The previous algorithm still involves computation of skyline. Ideally, we would like to avoid the computation as much as possible. We observed earlier that the skyline results for neighboring cells are different only due to the point on the shared grid line. For example, in Figure \ref{fig:onewayL}, $Sky(C_{1,2})$ and $Sky(C_{2,2})$ are different due to $p_6$, same for $Sky(C_{1,3})$ and $Sky(C_{2,3})$. Similarly, $Sky(C_{1,2})$ and $Sky(C_{1,3})$ are different due to $p_9$, same for $Sky(C_{2,2})$ and $Sky(C_{2,3})$. In this subsection, we observe an interesting property of the exact relationship between the skyline results of neighboring cells, and present a new $O(n^3)$ time algorithm utilizing this property for computing skyline for all cells.  Again, the merging of cells into skyline polyominos stays the same as the baseline.

\begin{figure}[htb]
 \centering
 \includegraphics[width=0.32\textwidth]{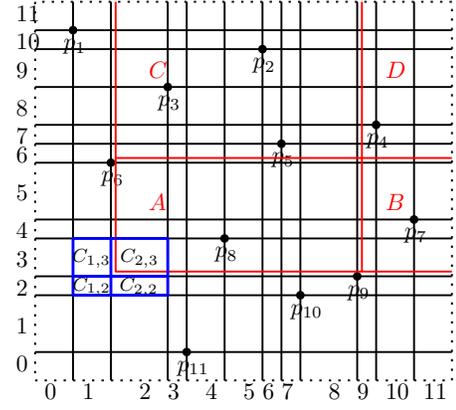}
 \captionsetup{font={scriptsize}}
 \caption{Scanning algorithm.}
 \label{fig:onewayL}
\end{figure}

\begin{theorem}\label{the:cell}
Given any skyline cell $C_{i,j}$ (except the ones that have a point as its upper right corner), and its right cell $C_{i+1,j}$, upper cell $C_{i,j+1}$, and upper right cell $C_{i+1,j+1}$, their skyline results have a relationship as follows.

$Sky(C_{i,j})=Sky(C_{i+1,j})+Sky(C_{i,j+1})-Sky(C_{i+1,j+1})$\footnote{multiset operation.}
\end{theorem}

\begin{proof}
Given a cell $C_{i,j}$, we define the following. ${p_R}$ (${p_C}$) denotes the point that lies on the upper (right) grid line of $C_{i,j}$. Range $A$ is the rectangle formed by the grid lines crossing $p_R$ and $p_C$ (excluding the two points). Range $B$ is the right rectangle of $A$. Range $C$ is the upper rectangle of $A$. And Range $D$ is the upper right rectangle of $A$. An example is shown in Figure \ref{fig:onewayL}.

Consider $C_{i,j}$'s upper right cell $C_{i+1,j+1}$, we denote $SkyP(A)$ as the set of points in range $A$ contributed to $Sky(C_{i+1,j+1})$.
And similarly for $SkyP(B)$, $SkyP(C)$, and $SkyP(D)$. Note that $SkyP(D)$ will be empty if $SkyP(A)$ is not empty which will dominate all points in $D$.

We can compute the skyline results of the four cells as follows.
\begin{displaymath}
Sky(C_{i,j})=\{p_R\}\cup \{p_C\}\cup SkyP(A)
\end{displaymath}
\begin{displaymath}
Sky(C_{i+1,j})=\{p_R\}\cup SkyP(A)\cup SkyP(C)
\end{displaymath}
\begin{displaymath}
Sky(C_{i,j+1})=\{p_C\}\cup SkyP(A)\cup SkyP(B)
\end{displaymath}
\begin{displaymath}
Sky(C_{i+1,j+1})=SkyP(A)\cup SkyP(B)\cup SkyP(C)\cup SkyP(D)
\end{displaymath}

Then we have:
\begin{displaymath}
Sky(C_{i+1,j})+Sky(C_{i,j+1})-Sky(C_{i+1,j+1})
\end{displaymath}
\begin{displaymath}
=(\{p_R\}\cup SkyP(A)\cup SkyP(C))
+(\{p_C\}\cup SkyP(A)\cup SkyP(B))
\end{displaymath}
\begin{displaymath}
-(SkyP(A)\cup SkyP(B)\cup SkyP(C)\cup SkyP(D))
\end{displaymath}
\begin{displaymath}
=\{p_R\}\cup \{p_C\}\cup SkyP(A) = Sky(C_{i,j})
\end{displaymath}
\end{proof}

\begin{example}
Given cell $C_{1,2}$ in Figure \ref{fig:onewayL}, $p_R$ is $p_9$ and $p_C$ is $p_6$. Consider the skyline result of its upper right cell $C_{2,3}$, we have $SkyP(A)=\{p_8\}$, $SkyP(B)=\emptyset$ as $p_7$ is dominated by $p_8$, $SkyP(C)=\{p_3\}$ as $p_2,p_5$ are dominated by $p_8$, and $SkyP(D)=\emptyset$ as $p_4$ is dominated by $p_8$. We have skyline result for the upper right cell $Sky(C_{2,3})=\{p_3,p_8\}$, the upper cell $Sky(C_{1,3})=\{p_6,p_8\}$, and the right cell  $Sky(C_{2,2})=\{p_3,p_8,p_9\}$. It is easy to see that the skyline for the given cell is $Sky(C_{1,2})= Sky(C_{2,2})+Sky(C_{1,3})-Sky(C_{2,3})=\{p_6,p_8,p_9\}$.
\end{example}

We note that the above property holds for all skyline cells except the ones that have a point as its upper right corner.  For these cells, their skyline is the upper right point because this point dominates all the upper right region. For example, in Figure \ref{fig:onewayL}, $Sky(C_{4,3})=\{p_8\}$ and $Sky(C_{6,6})=\{p_5\}$.

\begin{algorithm}[thb] \scriptsize\caption{The scanning algorithm for skyline diagram of quadrant skyline queries.} \label{Alg:luo}
\SetKwInOut{Input}{input}\SetKwInOut{Output}{output}

\Input{a set of $n$ points and skyline cells $C_{i,j}$.}
\Output{skyline of each skyline cell $Sky(C_{i,j})$.}

\For{i=0 to n}{
$Sky(C_{i,n})=\emptyset$\;
$Sky(C_{n,i})=\emptyset$\;
}

\For{i=n-1 to 0}{
    \For{j=n-1 to 0}{
    \If{there is a point $p$ on the upper right corner of $C_{i,j}$}{
    $Sky(C_{i,j})$=\{p\}\;
    }
    \Else{
    $Sky(C_{i,j})=Sky(C_{i+1,j})+Sky(C_{i,j+1})-Sky(C_{i+1,j+1})$\;}

}}

\end{algorithm}

Based on these properties, we present a scanning algorithm as shown in Algorithm \ref{Alg:luo}. The basic idea is to start from the top and rightmost cell, and scan the cells from the top down and right to left, then utilizing the property in Theorem 1 to compute the skyline for each cell. We first initialize the skyline results for the skyline cells on the top row and rightmost column to $\emptyset$ (Lines 1-3).  Then for each cell $C_{i,j}$, if there is a point $p$ on its upper right corner, we set $Sky(C_{i,j})=\{p\}$ (Line 7). Otherwise, we use $Sky(C_{i,j})=Sky(C_{i+1,j})+Sky(C_{i,j+1})-Sky(C_{i+1,j+1})$ to compute the skyline of $C_{i,j}$ (Line 9).

\partitle{Complexity}
There are $O(n^2)$ cells, each cell requires $O(n)$ time for multiset computation. Therefore, Algorithm \ref{Alg:luo} requires $O(n^3)$ time in total. We note that in practice, the time for multiset computation is much smaller than $n$. Thus the algorithm is much faster than the baseline algorithm in practice. Given a domain size $s$ for each dimension, the number of cells is bounded by $O(min(s^2,n^2))$, hence Algorithm \ref{Alg:luo} requires $O(min(s^2,n^2)n)$ time in total. The space complexity stays the same as the baseline algorithm which is $O(min(s^2,n^2)n)$.

\subsection{Sweeping Algorithm}

All previous algorithms involve computing skyline for each skyline cell (divided by the grid lines) and then merging them into skyline polyominos. Ideally, if we can find the skyline polyominos directly rather than combining the skyline cells, we can save the cost of computing skyline for each skyline cell. In this subsection, we show a sweeping algorithm that achieves this goal.

\begin{figure}[htb]
 \centering
 \includegraphics[width=0.32\textwidth]{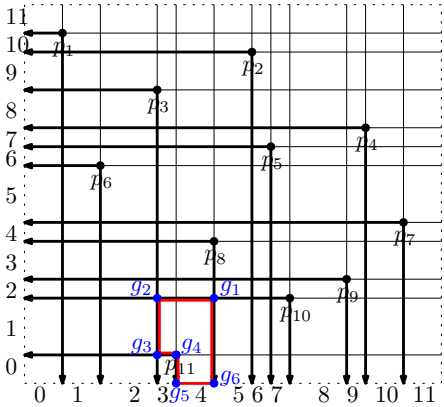}
 \captionsetup{font={scriptsize}}
 \caption{Sweeping algorithm.}
 \label{fig:onewayP}
\end{figure}

We observed previously that when we move from one cell to its right cell, the only change in the skyline result is caused by the point on the crossed grid line.  In fact, we can further observe that if the point on the crossed grid line lies below the cell, then the skyline result does not change at all. This is because we are only considering the points in the cell's upper right quadrant. For example, $C_{3,1}$ has skyline result $\{p_8, p_{10}\}$. When we move from $C_{3,1}$ to $C_{4,1}$ crossing point $p_{11}$, the skyline remains the same because $p_{11}$ is below the cells and does not affect the result. Similarly, when we move from one cell to its upper cell, if the point on the crossed grid line is to the left of the cells, the skyline result does not change either. In other words, each point only affects the skyline result of its lower and left cells, not its upper or right cells.  Motivated by this observation, instead of drawing grid lines over each point to divide the plane into skyline cells, we can draw two half-open grid lines starting from each point, one downward and another leftward. These $O(2n)$ grid line segments divide the plane into a set of polyominos, each containing one or more cells. Since we know that each point will not affect the skyline result of its upper and right cells, we can show that any query points in such formed polyominos have the same skyline results.  We have a theorem as follows.

\begin{theorem}
Given a set of points, if we draw two half-open grid lines starting from each point, one downward and another leftward, each polyomino formed by these $O(2n)$ lines is a skyline polyomino and all query points inside have the same first quadrant skyline query results.
\end{theorem}
\begin{proof}
Given a skyline polyomino formed by these half-open grid lines, if we consider the upper right corner query point for each of the skyline cells in the polyomino, they have the same set of points in their upper right quadrant, thus they have the same skyline results. We have shown earlier all points in the same skyline cell have the same quadrant skyline results, hence all query points in the same polyomino have the same first quadrant skyline results.
\end{proof}


\begin{algorithm}[thb]\scriptsize\captionsetup{font={scriptsize}}\caption{The sweeping algorithm for skyline diagram of quadrant skyline queries.} \label{Alg:pei}
\SetKwInOut{Input}{input}\SetKwInOut{Output}{output}

\Input{a set of $n$ points.}
\Output{skyline polyominos.}

/*compute all the intersection points and link them by left and right neighbors in Lines 4-10*/\;
sort the points in descending order on $y$-coordinate, $p_1$ ($p_n$) is the point with highest (lowest) $y$-coordinate\;
$p_1.left=(0,p_1[y])$\;
\For{i=2 to n}{
    insert $p_i$ into sorted queue $X$ by $x$-coordinate and its new index is $j$\;
    $p_i.left = (p_{j-1}[x],p_i[y])$\;
    $(p_{j-1}[x],p_i[y]).right = p_i$\;
    \For{j=i to 1 of sorted queue X}{
        $(p_{j-1}[x],p_i[y]).left =(p_{j-2}[x],p_i[y])$\;
        $(p_{j-2}[x],p_i[y]).right =(p_{j-1}[x],p_i[y])$\;
    }
}

/*similarly, we can compute the lower/upper neighbor of each intersection point*/\;
\For{each intersection point $g_0$}{
    $skymino_g$ = \{$g_0$\}; $g$= $g_0$\;
    $skymino_g$.append($g$.left); $g=g.left$\;
    \While{$g[x]!=g_0[x]$}{
     $skymino_g$.append($g$.lower); $g=g.lower$;
     $skymino_g$.append($g$.right); $g=g.right$;}
     \Return $skymino_g$\;

}
\end{algorithm}

While it is straightforward to visually see the skyline polyominos from the figure (e.g., Figure \ref{fig:onewayP}), we need to represent the skyline polyominos computationally by its vertices, which are the intersection points of the half-open grid lines including the points themselves. We now show how to compute the coordinates of these vertices and then how to find the vertices for each polyomino.

We observe that for each point $p$, its horizontal grid line only intersects with the vertical grid lines from its upper points, i.e., with larger y coordinates.  Hence, given a point $p(x,y)$, we can compute all the intersection points on its horizontal grid line as $g(x_j,y)$, where $x_j$ is the x coordinate from those points with larger y coordinates than $p$. For each intersection point, we record its left and right neighbor, so that we can retrieve the vertices for each polyomino. Similarly, for each point, we compute the intersection points on its vertical grid line, and record the lower and upper neighbor for each intersection point.  The detailed algorithm is shown in Algorithm \ref{Alg:pei}.

\begin{example}
For $p_{4}$ in Figure \ref{fig:onewayP}, its horizontal line intersects with the vertical lines of $p_2, p_3, p_1$, hence the intersection points on its horizontal line are $(p_2[x], p_4[y]), (p_3[x], p_4[y]), (p_1[x], p_4[y])$, and $(0,p_4[y])$.  For each point, it has a left/right and upper/lower neighbor, e.g., $(p_2[x], p_4[y]).right=p_4$.
\end{example}

Once all the intersection points are computed and linked by their left/right and lower/upper neighbors, we can retrieve the sequence of vertices for each polyomino.  We can see that each intersection point has a uniquely corresponding polyomino with the point as its upper right corner. Therefore, for each intersection point $g$, we find the sequence of vertices forming its corresponding polyomino.  The polyominos are either rectangles or half-rectangles with lower left side shaped like steps.
Hence we first retrieve $g$'s left neighbor.  We then repeatedly find the next lower neighbor and right neighbor until the right neighbor reaches the same $y$ coordinate as the original intersection point $g$.

\begin{example}
For the intersection point $g_1(p_8[x],p_{10}[y])$ in Figure \ref{fig:onewayP}, we first find its left vertex $g_2(p_3[x],p_{10}[y])$.  We then find the lower vertex $g_3(p_3[x],p_{11}[y])$, and the right vertex $g_4(p_{11}[x],p_{11}[y])$ in the first iteration.  Because $g_4$ is not meeting the grid line at $g_1$ yet, it continues to find the next lower vertex $g_5(p_{11}[x],0)$ and the right vertex $g_6(p_8[x],0)$. Now the algorithm stops as $g_6$ reaches the $y$ grid line of $g_1$.  The sequence of vertices for the skymino corresponding to $g_1$ is $g_1, g_2, g_3, g_4, g_5, g_6$.
\end{example}

\partitle{Complexity}
The computation of intersection points requires $O(n^2)$ time. Because each grid line segment between two neighboring intersection points will be used at most twice for constructing skyminos, the skymino constructing step requires $O(n^2)$ time.  Therefore, Algorithm \ref{Alg:pei} requires $O(n^2)$ time.
Given a domain size $s$ for each dimension, the number of intersection points is bounded by $O(min(s^2,n^2))$, hence Algorithm \ref{Alg:pei} requires $O(min(s^2,n^2))$ time.  The space complexity stays the same as the baseline algorithm which is $O(min(s^2,n^2)n)$.

\subsection{Extension to High-dimensional Space}

In this subsection, we show how to adapt the baseline algorithm as well as the directed skyline graph algorithm and scanning algorithm from two dimensions to high dimensions. The sweeping algorithm, although provides the best performance on correlated dataset for two-dimensional space, it can not be easily extended to high-dimensional space and we leave its extension to future work.

\subsubsection{Baseline Algorithm}
We can construct $O(n^d)$ skyline (hyper) cells and easily see that each cell has the same skyline. For each cell, we find those points that lie on its first orthant (the counterpart of quadrant in high-dimensional space) and then use $O(n\log^{d-1}n)$ skyline algorithm to compute the skyline results for each cell. The cells with the same results are merged into polyminos.

\partitle{Complexity}
We have $O(n^d)$ skyline cells, and each cell requires $O(n\log^{d-1}n)$ time for finding the skyline because there is no monotonic property in high-dimensional space. The merging phase (which is the same for all algorithms) requires $O(n^d)$ time for searching in $d$-dimensional space. Thus, the baseline algorithm requires $O(n^{d+1}\log^{d-1}n)$ time. Since we have $O(n^d)$ skyminos and each skymino requires $O(n)$ space, the space complexity is $O(n^{d+1})$.  The above analysis assumes unlimited domain.  Given limited domain size $s_i$ for $i^{th}$ dimension, the number of skyline cells or skyminos is bounded by $O(min(\prod_{i=1}^d s_i,n^d))$. Hence, the time complexity is $O(min(\prod_{i=1}^d s_i,n^d)n\log^{d-1}n)$ and the space complexity is $O(min(\prod_{i=1}^d s_i,n^d)n)$. We note that all the high-dimensional algorithms have the same space complexity due to the same output structure.

\subsubsection{Directed Skyline Graph Algorithm}
Directed skyline graph in high-dimensional space can be constructed in $O(n^2)$ time \cite{DBLP:journals/pvldb/LiuXPLZ15}. Given directed skyline graph, the algorithm for high dimensions is exactly the same as the algorithm for two dimensions.

\partitle{Complexity}
Similar to the two-dimensional case, each ``row'' requires $O(n^2)$ to update the links and we have $O(n^{d-1})$ rows in $d$ dimensions. Thus, directed skyline graph algorithm requires $O(n^{d+1})$ time. Given domain size $s_i$ for $i^{th}$ dimension, the number of links is bounded by $O((min\{\prod_{i=1}^d s_i,n\})^2)$ and the number of rows is bounded by $O(\prod_{i=1}^{d-1} s_i)$, hence the complexity becomes $O(\prod_{i=1}^{d-1} s_i min\{\prod_{i=1}^d s_i,n\})^2)$. The space complexity stays the same as the baseline algorithm which is $O(min(\prod_{i=1}^d s_i,n^d)n)$.

\subsubsection{Scanning Algorithm}
Similar to Theorem \ref{the:cell} of Scanning algorithm in two dimensions, we have a relationship for high dimensions as follows.
\begin{displaymath}
Sky(C_{D_1,...,D_d})=Skyline(\sum^{\#~ of~ ``D_i+1''~ is~ odd} Sky(C_{D_1+\{0,1\},...,D_d+\{0,1\}})
\end{displaymath}
\begin{displaymath}
-\sum^{\#~ of~ ``D_i+1''~ is~ even} Sky(C_{D_1+\{0,1\},...,D_d+\{0,1\}}))
\end{displaymath}
where $D_i$ is the $i^{th}$ dimensional coordinate of skyline cell $C_{D_1,D_2,...,D_d}$, $+\{0,1\}$ means the coordinates either add $1$ or $0$ due to the neighbor relationship, and all operations are multiset operation.
For example in three dimensional space, $Sky(C_{D_1,D_2,D_3})=Skyline(Sky(C_{D_1+0,D_2+0,D_3+1})+Sky(C_{D_1+0,D_2+1,D_3+0})+Sky(C_{D_1+1,D_2+0,D_3+0})+Sky(C_{D_1+1,D_2+1,D_3+1})-Sky(C_{D_1+0,D_2+1,D_3+1})-Sky(C_{D_1+1,D_2+0,D_3+1})-Sky(C_{D_1+1,D_2+1,D_3+0}))$. The proof can be derived similar to Theorem \ref{the:cell}.

\partitle{Complexity}
We have $O(n^d)$ skyline cells, each cell requires $O(n\log^{d-1}n)$ time to do the multiset operations. Thus, scanning algorithm requires $O(n^{d+1}\log^{d-1}n)$ time.  We note that in practice, the number of remaining points is much smaller than n. Thus the algorithm is much faster than the baseline algorithm in practice. Given domain size $s_i$ for $i^{th}$ dimension, the number of skyline cells is bounded by $O(min(\prod_{i=1}^d s_i,n^d))$, hence the time complexity becomes $O(min(\prod_{i=1}^d s_i,n^d)n\log^{d-1}n)$. The space complexity stays the same as the baseline algorithm which is $O(min(\prod_{i=1}^d s_i,n^d)n)$.


\section{Skyline Diagram of Dynamic Skyline}\label{sec:dynamic}

In this section, we study algorithms for skyline diagram of dynamic skyline in two dimensions. They can be extended to high dimensions similar to skyline diagram of quadrant/global skyline.  We first present a baseline algorithm and define an important notion of skyline subcell.  Then based on the observation that dynamic skyline query result is a subset of global skyline, we present an improved subset algorithm utilizing the skyline diagram of global skyline. Finally, based on the relationship of the skyline results of neighboring subcells, we present a scanning algorithm with improved complexity.

\subsection{Baseline Algorithm}

Similar to the skyline diagram of quadrant/global skyline, we can first find small regions that are guaranteed to have the same dynamic skyline, and then merge them to form skyline polyominos.

\partitle{Skyline Subcell}
In skyline diagram of quadrant/global skyline, each point contributes a horizontal and vertical grid line to divide the plane into skyline cells which are guaranteed to have the same result for quadrant skyline queries. For dynamic skyline, all points will be mapped to the first quadrant with respect to the query point and may dominate the points which are otherwise global skyline points.  Hence the points in the skyline cell are not guaranteed to have the same dynamic skyline. Therefore, to account for mapped points, in addition to the grid lines over each point, we draw a vertical and horizontal bisector line between each pair of points.  In total, we have $O{n\choose 2}$ horizontal lines and $O{n\choose 2}$ vertical lines which leads to $O({n\choose 2}^2)$ regions. Figure \ref{fig:dyAlg} shows an example with 4 points. The $4 \choose 2$ bisector lines between each pair of points and the 4 grid lines over each point divide the plane into $121$ regions. We can see that these regions are guaranteed to have the same dynamic skyline, since there are no points or mapped points in each of these regions that would change the dominance relationship of the points. To distinguish with skyline cell for quadrant/global skyline, we name these regions skyline subcells for dynamic skyline. 

\begin{figure}[!htb]
 \centering
 \includegraphics[width=0.45\textwidth]{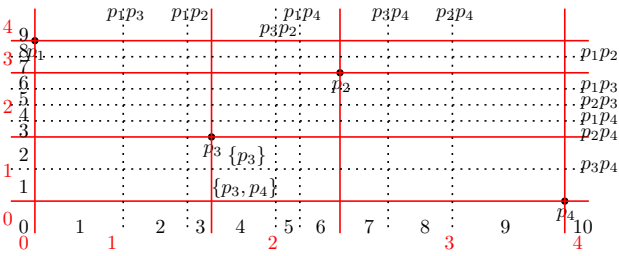}
 \captionsetup{font={scriptsize}}
 \caption{Skyline subcells for dynamic skyline (solid grid lines for cells and dotted lines for subcells).}
 \label{fig:dyAlg}
\end{figure}

\begin{definition}(\textbf{Skyline Subcell}).
The vertical and horizontal bisectors of each pair of points divide the plane into skyline subcells. Any query points in the same skyline subcell have the same dynamic skyline.
\end{definition}

\begin{algorithm}[thb] \scriptsize\caption{The baseline algorithm for skyline diagram of dynamic skyline.} \label{Alg:dynaBasic}
\SetKwInOut{Input}{input}\SetKwInOut{Output}{output}

\Input{skyline subcells $SC_{i,j}$.}
\Output{skyline of each skyline subcell $Sky(SC_{i,j})$.}

\For{i=0 to mx}{
    \For{j=0 to my}{
        \For{k=1 to n}{
            $p_k[x]'=|p_k[x]-SC_{i,j}[x]|$\;
            $p_k[y]'=|p_k[y]-SC_{i,j}[y]|$\;}
        employ skyline algorithm on $p_k'$ for $k=1,...,n$ to compute the skyline as the output of $SC_{i,j}$\;
}}
\end{algorithm}

\partitle{Finding skyline for each skyline subcell}
Once we have the skyline subcells, we can compute the skyline for each subcell. The baseline algorithm is straightforward and similar to the skyline computation for skyline cells as shown in Algorithm \ref{Alg:dynaBasic}. For each subcell $SC_{i,j}$, it first maps all the points to the first quadrant with respect to the subcell (Lines 4-5).  It then computes the skyline of the mapped points.

\partitle{Complexity}
Since skyline can be computed in $O(n)$ time if the points are sorted on one dimension, and there are $O(n^4)$ subcells, the entire algorithm (Algorithm \ref{Alg:dynaBasic}) can be finished in $O(n^5)$. Similarly, the space complexity is $O(n^5)$. We note that the remaining algorithms in this section have the same space complexity due to the same output structure. In practice, given a limited domain size $s$ for each dimension, the number of subcells is bounded by $O(min(s^2,n^4))$ because most of the bisector lines are coincident.  Hence the time and space complexity becomes $O(min(s^2,n^4)n)$.

\subsection{Subset Algorithm}

As we discussed earlier, the mapped points may dominate additional points that would have been global skyline points. As a result, the dynamic skyline of each subcell $SC_{i,j}$ is a subset of the global skyline of the skyline cell it belongs to. For example, in Figure \ref{fig:dyAlg}, $Sky(SC_{3,1})$ is a subset of $Sky(C_{1,1})$. Therefore, we can first use the algorithms in the previous section to compute the global skyline of the skyline cells, and then compute the dynamic skyline of each subcell from this set rather than the entire $n$ points. The detailed algorithm is shown in Algorithm \ref{Alg:dynaSubset} which is very similar to the baseline algorithm. The only difference is we just need to consider the output of global skyline results of each skyline cell rather than the entire $n$ points.

\begin{algorithm}[thb]\scriptsize \caption{The subset algorithm for skyline diagram of dynamic skyline.} \label{Alg:dynaSubset}
\SetKwInOut{Input}{input}\SetKwInOut{Output}{output}

\Input{global skyline result of each skyline cell $Sky(C_{i,j})$.}
\Output{dynamic skyline result of each skyline subcell $Sky(SC_{i,j})$.}


\For{k=0 to mx}{
    \For{l=0 to my}{
        find $C_{i,j}$ such that $SC_{k,l}\in C_{i,j}$\;
        $Sky(SC_{k,l})$ = dynamic skyline of the points in $Sky(C_{i,j})$
}}
\end{algorithm}
\vspace{-2em}
\partitle{Complexity} Although the worst case time complexity is the same as the baseline algorithm $O(n^5)$, on average, the number of skyline for $n$ points is only $O(\log n)$. Therefore, the amortized time complexity of the subset algorithm is reduced to $O(n^4\log n)$. We will show that the subset algorithm is indeed significantly faster than the baseline algorithm in practice in Section \ref{sec:experiments}. Again, given a limited domain size $s$ for each dimension, the number of subcells is bounded and hence the time and space complexity is $O(min(s^2,n^4)n)$.

\subsection{Scanning Algorithm}

The baseline and subset algorithms compute the skyline for each subcell from scratch.  To further improve the efficiency, in this subsection, we propose an incremental scanning algorithm based on the relationship of the dynamic skyline results of neighboring subcells.  This is due to the observation that as we move from one subcell to its neighboring subcell on the right, the only difference of the skyline result is caused by the two points that contributed to the bisector line between the two subcells. We just need to consider these two points in addition to the skyline result of the previous subcell. For example in Figure \ref{fig:dyAlg}, $Sky(SC_{4,2})=\{p_3\}$, for $SC_{4,1}$, we only need to check $\{p_3\}\cup \{p_3,p_4\}=\{p_3,p_4\}$. Because $p_3,p_4$ cannot dominate each other, therefore, $Sky(SC_{4,1})=\{p_3,p_4\}$. So similar to the scanning algorithm for quadrant skyline queries, we first compute $Sky(SC_{0,0})$ for the lower left subcell. We then scan the subcells from left to right on the first row and compute the skyline incrementally.  We then compute each of the remaining rows from bottom up.  The detailed algorithm is shown in Algorithm \ref{Alg:dynaScan}.

\begin{algorithm}[thb] \scriptsize \caption{The scanning algorithm for skyline diagram of dynamic skyline.} \label{Alg:dynaScan}
\SetKwInOut{Input}{input}\SetKwInOut{Output}{output}

\Input{a set of $n$ points and skyline subcells $SC_{i,j}$.}
\Output{skyline of each skyline subcell $Sky(SC_{i,j})$.}

employ skyline algorithm to compute the skyline of Subcell $SC_{0,0}$\;
\For{i=1 to mx}{
        $Sky(SC_{i,0})$ = $Sky(SC_{i-1,0})\bigcup the~ points$ contributing to the bisectors between $SC_{i-1,0}$ and $SC_{i,0}$\;
}

\For{i=0 to mx}{
    \For{j=1 to my}{
         $Sky(SC_{i,j})$ = skyline from $Sky(SC_{i,j-1})\bigcup the~ points$ contributing to the bisectors between $SC_{i,j-1}$ and $SC_{i,j}$\;
}}

\end{algorithm}

The key step in the above algorithm is to compute the updated skyline given the skyline result of the previous cell and the new points contributing to the bisectors (Line 3 and Line 8).  When adding a new point, there are two cases: 1) the new point becomes a skyline point which may dominate some existing skyline points, or 2) the new point is dominated by existing skyline points. To determine if the new point is dominated by existing skyline points, we can do a binary search to find the skyline point $p_i$ such that $p_i[x]\leq p[x]$ and $p[x]\leq p_{i+1}[x]$. If $p_i[y]\geq p[y]$, the new point is a skyline point, otherwise, the new point is dominated by $p_i$. This procedure can be finished in $O(\log n)$ time.  If the new point is a skyline point, we need to remove those points dominated by the new point. If we sort the skyline points in ascending order on $x$-coordinate and descending order on $y$-coordinate, we can delete those points in $O(\log n)$ time.

\partitle{Complexity}
Since the computation of updated skyline for each subcell only costs $O(\log n)$ time, and there are $O(n^4)$ subcells, the overall worst case time complexity for the scanning algorithm is $O(n^4\log n)$. Again, given a limited domain size $s$ for each dimension, the number of subcells is bounded and hence the time complexity is $O(min(s^2,n^4)\log n)$.  The space complexity is the same as the baseline algorithm which is $O(min(s^2,n^4)n)$.

\section{Approximate Skyline Diagram}

The drawback of skyline diagram is the high space cost. In this section, we propose an approximate skyline diagram to significantly reduce the space cost. The key idea of the approximate skyline diagram is to allow nearby skyline cells that have different but similar skyline results to be merged into one skyline polyomino in order to reduce the number of skyline polyominos and hence reduce the space cost. However, this could sacrifice the precision of the skyline query result, i.e., each skyline polyomino now has the union of the skyline points of each skyline cell within the skyline polyomino, which is a superset of the actual skyline result given any query point within the skyline polyomino. A query user then needs to process this superset to find the actual result. The more skyline cells we merge, the higher precision we sacrifice, since the size of the skyline union can be much larger than the actual number of skyline points of each single skyline cell. The extreme case for the approximate skyline diagram is that we combine all the skyline cells into one skyline polyomino. In this case, the skyline diagram is not useful because it defeats the purpose of skyline query by returning the union of the skyline points of all skyline cells, whose size can be as large as $n$ in the worst case.

Therefore, we propose to use $n_h$ Horizontal Partitioning Lines (HPLs) and $n_v$ Vertical Partitioning Lines (VPLs) to partition the whole skyline diagram into $(n_h-1)(n_v-1)$ skyline polyominos such that each skyline polyomino contains at most $\delta$ skyline points. Parameter $\delta$ guarantees that the user only needs to consider at most $\delta$ skyline points/choices so that the approximation is controlled. Our optimization goal here is to find the smallest $n_h+n_v$ which ``approximately'' minimizes the storage cost.

\begin{definition}[Approximate Skyline Diagram Problem]\label{def:ASDP}
Given a skyline diagram with $n\times n$ skyline cells (without merging skyline cells into skyline polyominos), we partition the skyline diagram into $(n_h-1)(n_v-1)$ skyline polyominos with the minimum number of $HPLs$ $n_h$ plus number of $VPLs$ $n_v$ such that each skyline polyomino contains at most $\delta$ skyline points.
\end{definition}

We first prove the NP-hardness of the approximate skyline diagram problem. Therefore, it is unlikely that there are efficient algorithms for solving this problem exactly. We then propose two heuristic algorithms to efficiently compute the approximate skyline diagram.

\subsection{NP-hardness of the Approximate Skyline Diagram}
In this subsection, we prove that the approximate skyline diagram problem is NP-hard by showing that the rectangular partitioning problem is polynomial time reducible to the approximate skyline diagram problem.

\begin{definition}[Rectangular Partitioning Problem]\cite{Forsmann07}\cite{DBLP:conf/icdt/MuthukrishnanPS99}
Given a set of non-overlapping rectangles $R_1, R_2, ..., R_n$, we partition the plane into tiles with the minimum number of rows plus columns such that each resulting tile intersects (adjacent boundary touch is not considered as an intersection) at most $\delta$ rectangles.
\end{definition}

Similarly, we have the decision version of the rectangular partitioning problem as follows,

\begin{definition}[Decision Version of the Rectangular Partitioning Problem]
Given a set of non-overlapping rectangles $R_1, R_2, ..., R_n$ and the values $n_h'$ and $n_v'$, is there a partitioning $(n_h',n_v')$ of the plane such that each tile intersects at most $\delta$ rectangles, where $n_h'$ is the number of HPLs and $n_v'$ is the number of VPLs.
\end{definition}

It was shown in \cite{Forsmann07} that the decision version of the rectangular partitioning problem is NP-hard even when $\delta=1$. Similarly, we have the decision version of the approximate skyline diagram problem as follows.

\begin{definition}[Decision Version of the Approximate Skyline Diagram Problem]
Given a skyline diagram with $n\times n$ skyline cells, each cell $C_{i,j}$ of the skyline diagram having a set of points $Sky(C_{i,j})$, and the values $n_h$ and $n_v$, is there a partitioning $(n_h,n_v)$ of the skyline diagram such that each resulting skyline polyomino contains at most $\delta$ unique symbols.
\end{definition}

\begin{theorem}
The approximate skyline diagram problem is NP-hard even when $\delta=1$.
\end{theorem}

\begin{figure}
    \centering
    \includegraphics[width=0.5\textwidth]{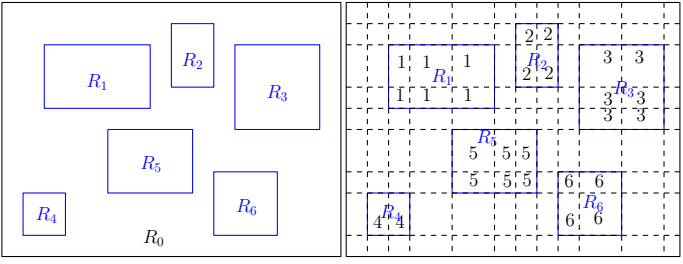}
    \caption{An  instance of mapping between Rectangular Partitioning Problem and Approximate Skyline Diagram Problem.}
   \label{fig:reduction}
\end{figure}
\begin{proof}
To reduce the rectangular partitioning problem to the approximate skyline diagram problem, our construction is shown as follows. Given a set of non-overlapping rectangles $R_1, R_2, ..., R_n$, we find a large enough rectangle $R_0$ which includes all those non-overlapping rectangles $R_1, R_2, ..., R_n$. That is, we use $R_0$ to fix the boundary. We extend four boundary line segments of each rectangle $R_i$ until to the boundary of $R_0$. Then rectangle $R_0$ is separated into small orthogonal grids. For a grid $G_{i,j}$, if it is located in $R_i$, then we set the symbol set $S_{i,j}$=``i''. If it is not in any rectangle, then we set the symbol set $S_{i,j}=\emptyset$. An example is shown in Figure ~\ref{fig:reduction}.

We show the reduction as follows.

(1)$\longrightarrow$ If there exists a partitioning $(n_h',n_v')$ such that each grid intersects at most one of the rectangles, and the lines in partitioning $(n_h',n_v')$ overlap with the grid lines, then partitioning $(n_h,n_v)$ exactly equals to partitioning $(n_h',n_v')$. Otherwise, if the VPL does not overlap with any vertical grid lines, we move this VPL to overlap with the nearest grid line on its left. Similarly, if the HPL does not overlap with any horizontal grid lines, we move this HPL to overlap with the nearest grid line above it. Let partitioning $(n_h,n_v)$ be the partitioning $(n_h',n_v')$ after moving.

Claim: each skyline polyomino in partitioning $(n_h,n_v)$ contains at most one unique symbol. Before moving, each tile intersects at most one rectangle. Hence, the tile after moving can only touch the boundary of the other rectangles without including any interior part of other rectangles since we only move those partitioning lines to their nearest grid lines and all the boundary lines of the rectangles are grid lines.

(2)$\longleftarrow$ If there exists a partitioning $(n_h,n_v)$ of skyline diagram such that each resulting skyline polyomino contains at most one unique symbol, and let partitioning $(n_h',n_v')$ be the partitioning $(n_h,n_v)$. Each tile in the partitioning $(n_h',n_v')$ intersects at most one rectangle. Otherwise, the skyline polyomino will have more than one unique symbol.
\end{proof}

\subsection{Algorithms for Computing the Approximate Skyline Diagram}
In this subsection, we show two heuristic algorithms, Bottom-Up Merging algorithm and Top-Down Partitioning algorithm, to efficiently compute the approximate skyline diagram in two dimensional space. We note that both two algorithms are applicable for skyline diagram of global and dynamic skyline, and they are straightforwardly extendable for high dimensional space.

\subsubsection{Bottom-Up Merging Algorithm}
In this subsection, we show a bottom-up merging algorithm to compute the approximate skyline diagram. The general idea is to merge as many skyline cells that satisfy the upper limit of $\delta$ skyline points in each skyline polyomino as possible. For each row, we scan each cell from left to right, and we find the maximum number of skyline cells that the number of points in the union of the skyline points is $\leq \delta$. We find the smallest number $n_{smallest}$ among all the $n$ rows and set the $(s_{smallest}+1)^{th}$ grid line as the second VPL, where we consider the first vertical grid line as the first VPL. In this case, we can guarantee that we do not need to partition the space between the first VPL and the second VPL anymore. Similarly, we can find all VPLs. For each row, we merge the skyline cells between each two adjacent VPLs and get a new skyline diagram with $n$ rows and $n_v-1$ columns, where $n_v$ is the number of VPLs. Using the similar method, we can find all the HPLs.

\begin{algorithm}[thb]\scriptsize \caption{Merging-based bottom-Up algorithm.} \label{Alg:bottomUp}
\SetKwInOut{Input}{input}\SetKwInOut{Output}{output}

\Input{A skyline diagram with $n\times n$ skyline cells and parameter $\delta$.}
\Output{An approximate skyline diagram.}

currentVL=1\;

\If{currentVL $\leq$ n}{
    \For{i=1 to n}{
        tempUnion[i]=$\emptyset$\;
        tempVL[i]=currentVL\;
        \For{j=currentVL to n}{
            $tempUnion[i]=\bigcup \{tempUnion[i],Sky(C_{currentVL,j})\}$\;
            \If{$|tempUnion[i]|>\delta$}{
            break;}
            tempVL[i]=j-1\;}
        find the smallest value $SV$ from all tempVL[i], i=1,2,...,n\;}
    currentVL=$SV$\;
    add the $SV^{th}$ vertical line to partitioning line pool\;
    }

denote the number of vertical partitioning lines as $n_v$\;
we have vertical partitioning lines $VPL_1,VPL_2,...,VPL_{n_v}$, where $VPL_1=1$ and $VPL_{n_v}=n+1$\;
currentHL=1\;

\For{i=1 to n}{
    \For{j=1 to $n_v$-1}{
        $Sky(C_{i,j})=\emptyset$\;
        \For{k=$VPL_j$ to $VPL_{j+1}$-1}{
            $Sky(C_{i,j})=\bigcup \{Sky(C_{i,j}), Sky(C_{i,k})\}$\;}}}
we have a new skyline diagram with $n$ rows and $n_v$-1 columns\;
similar to Lines 1-15, we have horizontal partitioning lines $HPL_1,HPL_2,...,HPL_{n_h}$, where $HPL_1=1$ and $HPL_{n_h}=n+1$\;

\For{i=1 to $n_v$-1}{
    \For{j=1 to $n_h$-1}{
        merge all the skyline cells lying between grid vertical line $VPL_i, VPL_{i+1}$ and grid horizontal line $HPL_j,HPL_{j+1}$\;}}

\end{algorithm}

The detailed algorithm is shown in Algorithm \ref{Alg:bottomUp}, we find all VPLs in Lines 1-15. For each row, we merge the skyline cells between the adjacent VPLs in Lines 17-21. We find all HPLs in Line 23. Finally, we merge the skyline cells between the HPLs and VPLs in Lines 24-26. The approximate skyline diagram has $n_v-1$ rows and $n_h-1$ columns, and each skyline polyomino contains at most $\delta$ points.

\subsubsection{Top-Down Partitioning Algorithm}
When $\delta$ is large, it is time-consuming to merge the skyline cells one by one. In this subsection, we show a top-down partitioning algorithm to compute the approximate skyline diagram, which is desired when $\delta$ is large. The general idea is to partition the plane into the minimum number of skyline polyominos that satisfy the upper limit of $\delta$ skyline points in each skyline polyomino. We assume $n_v=n_h=p$. We first guess that we only need $p'=2$ VPLs and HPLs, i.e., we partition the skyline diagram into one skyline polyomino. And then we check if each skyline polyomino contains more than $\delta$ points. If it does, we double $p'$ and check again until each polyomino contains $\leq \delta$ points. That is, we find a $p'$ such that $p'$ satisfies the requirement but $p'/2$ not. We then use binary search to find the exact $p$ between $p'/2$ and $p'$ such that $p$ satisfies the requirement but $p-1$ not. The remaining problem is how to ``equally'' partition the skyline diagram. We take the number of the skyline points in each skyline cell as its weight, and then we have the weights for each row and each column. We partition the skyline diagram based on the weights of the rows and the columns.

\begin{algorithm}[thb]\scriptsize \caption{Top-Down Partitioning algorithm.} \label{Alg:topDown}
\SetKwInOut{Input}{input}\SetKwInOut{Output}{output}

\Input{A skyline diagram with $n\times n$ skyline cells and parameter $\delta$.}
\Output{An approximate skyline diagram.}

\For{i=1 to n}{
    \For{j=1 to n}{
        $W(C_{i,j})=|Sky(C_{i,j})|$\;}}

\For{i=1 to n}{
    $W(R_i)=0$\;
    \For{j=1 to n}{
        $W(R_i)=W(R_i)+W(C_{i,j})$\;}}

similar to Lines 4-7, we have $W(C_i)$ for all columns, where i=1,2,...,n\;

p'=2\;
partition the skyline diagram into $(p'-1)^2$ skyline polyominos equally based on the weights of the rows and the columns\;

\While{one of the skyline polyomino contains more than $\delta$ points}{
    p'=2p'\;
    partition the skyline diagram into $(p'-1)^2$ skyline polyominos equally based on the weights of the rows and the columns\;}

use binary search to find the exact $p$ between $p'/2$ and $p'$ such that $p$ satisfies the requirement but $p-1$ not.

partition the skyline diagram into $(p-1)^2$ skyline polyominos equally based on the weights of the rows and the columns\;

use the similar method of Lines 24-26 in Algorithms \ref{Alg:bottomUp} to compute the final approximate skyline diagram\;

\end{algorithm}

The detailed algorithm is shown in Algorithm \ref{Alg:topDown}. In Lines 1-3, we compute the weight for each skyline cell. We compute the weights for the rows and the columns in Lines 4-8. We guess that we only need two VPLs and two HPLs in Line 9 and check if the approximate skyline diagram satisfies the requirement in Line 10. If it does not, we double the number of partitioning lines $p'$ until the approximate skyline diagram satisfies the requirement in Line 12. However, the exact number of partitioning lines $p$ should be a value between $p'/2$ and $p'$. Therefore, we use binary search to find the exact $p$ in Line 14. We partition the skyline diagram into $(p-1)^2$ skyline polyominos equally based on the weights and get the final approximate skyline diagram in Line 16.

\begin{figure*}[!htb]
\centering
\subfigure[\scriptsize{time cost of CORR}]{
\begin{minipage}[b]{0.22\textwidth}
\includegraphics[width=1.1\linewidth]{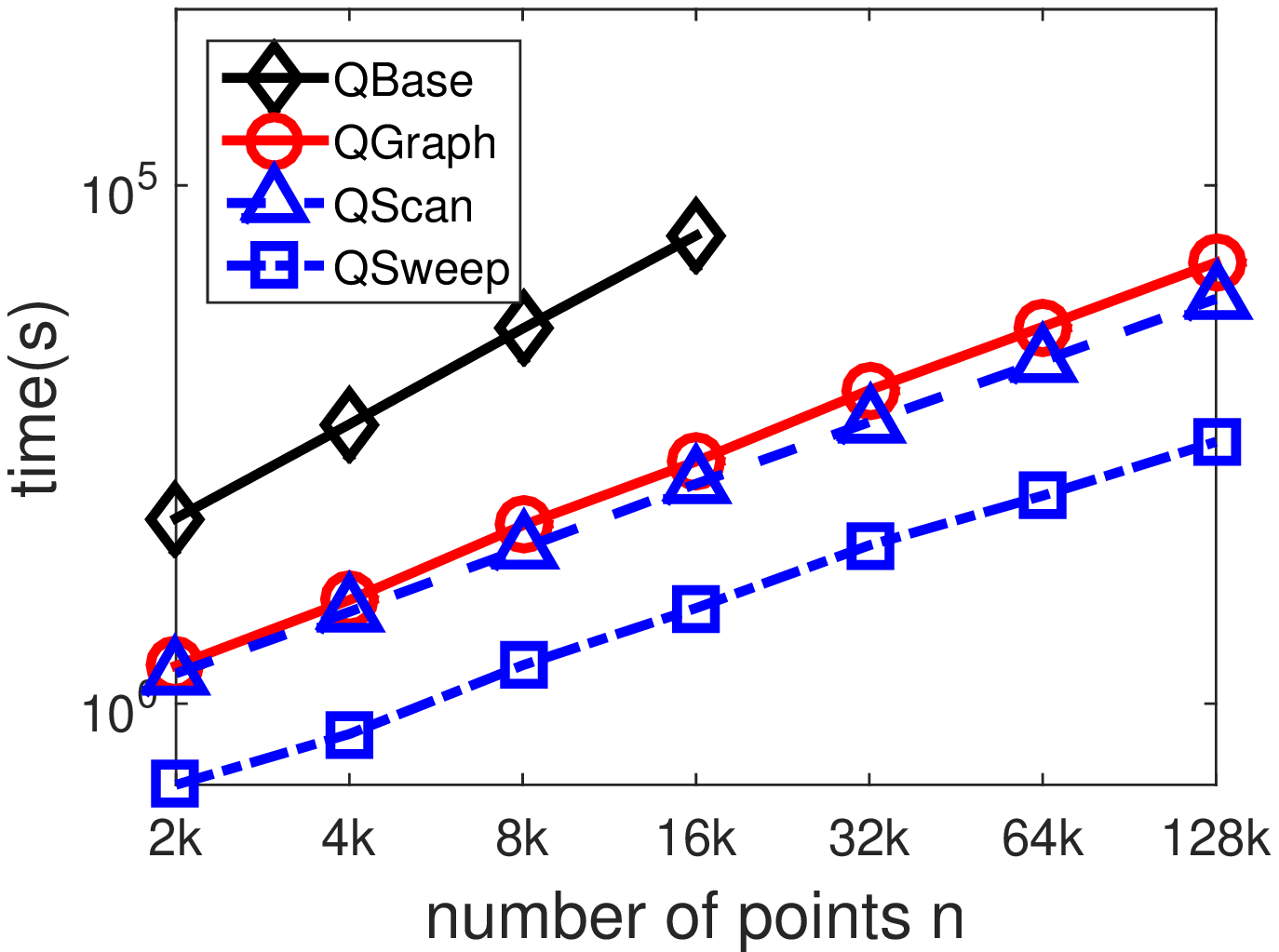}
\end{minipage}
}
\subfigure[\scriptsize{time cost of INDE}]{
\begin{minipage}[b]{0.22\textwidth}
\includegraphics[width=1.1\linewidth]{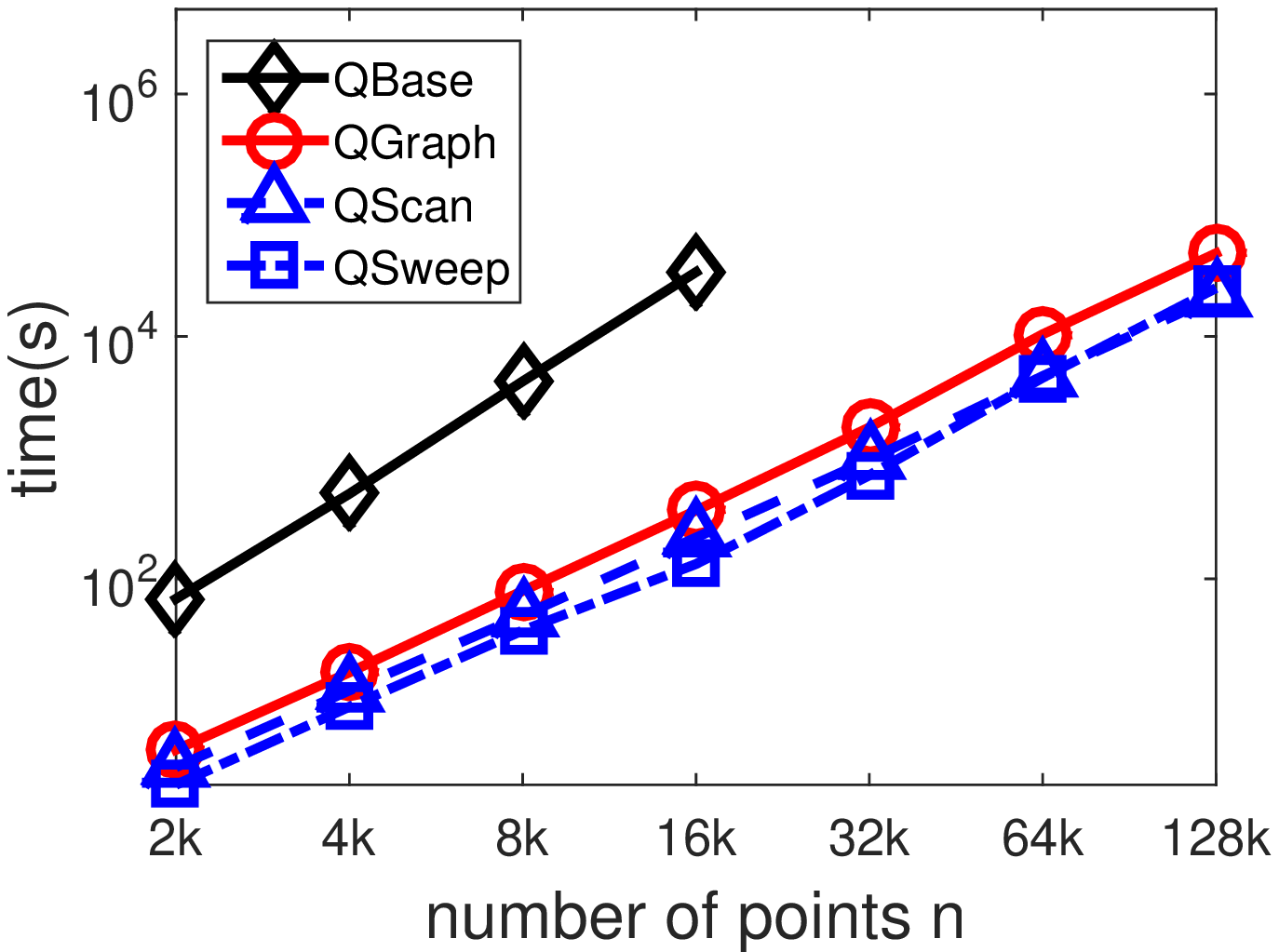}
\end{minipage}
}
\subfigure[\scriptsize{time cost of ANTI}]{
\begin{minipage}[b]{0.22\textwidth}
\includegraphics[width=1.1\linewidth]{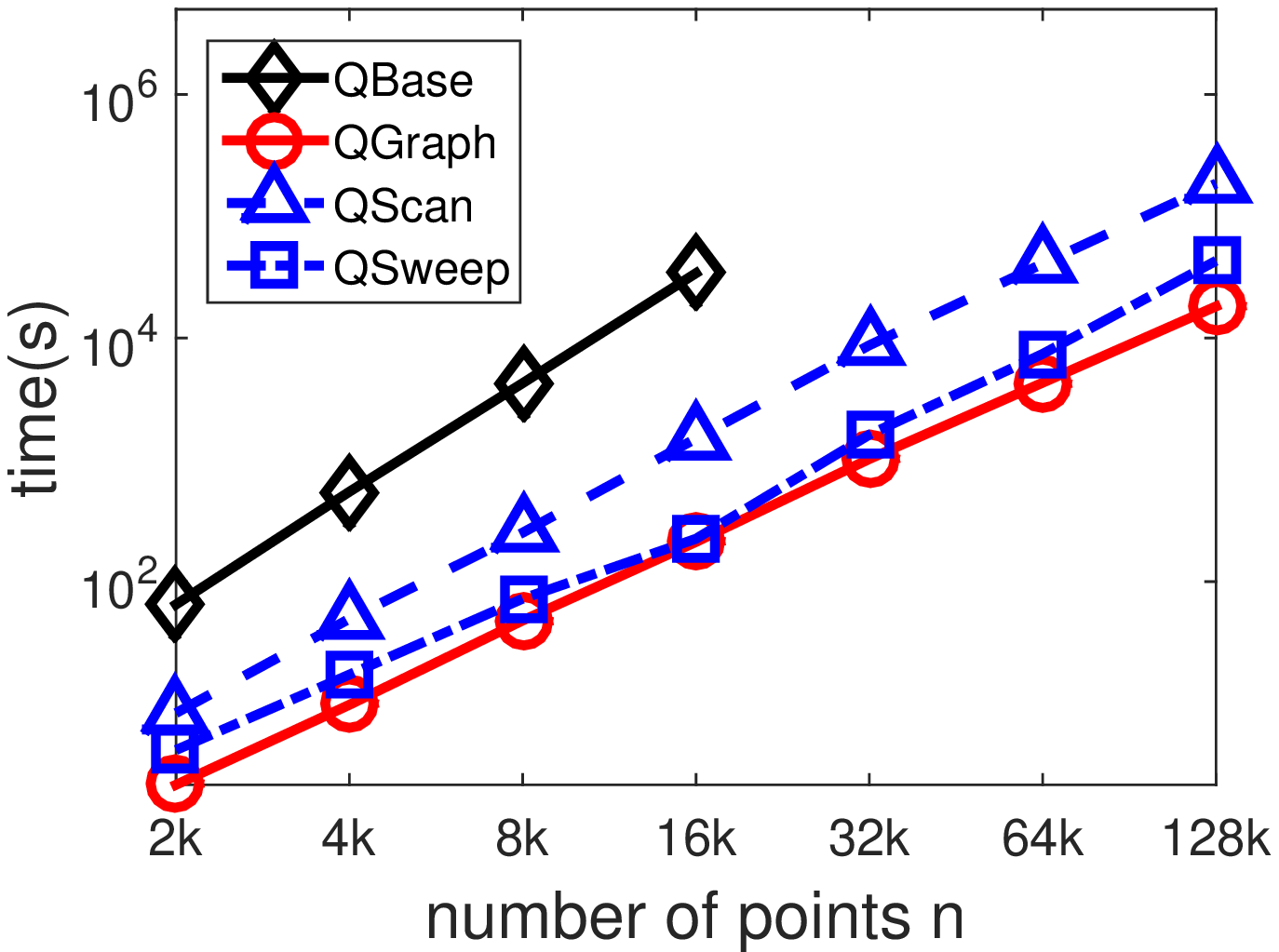}
\end{minipage}
}
\subfigure[\scriptsize{time cost of NBA}]{
\begin{minipage}[b]{0.22\textwidth}
\includegraphics[width=1.1\linewidth]{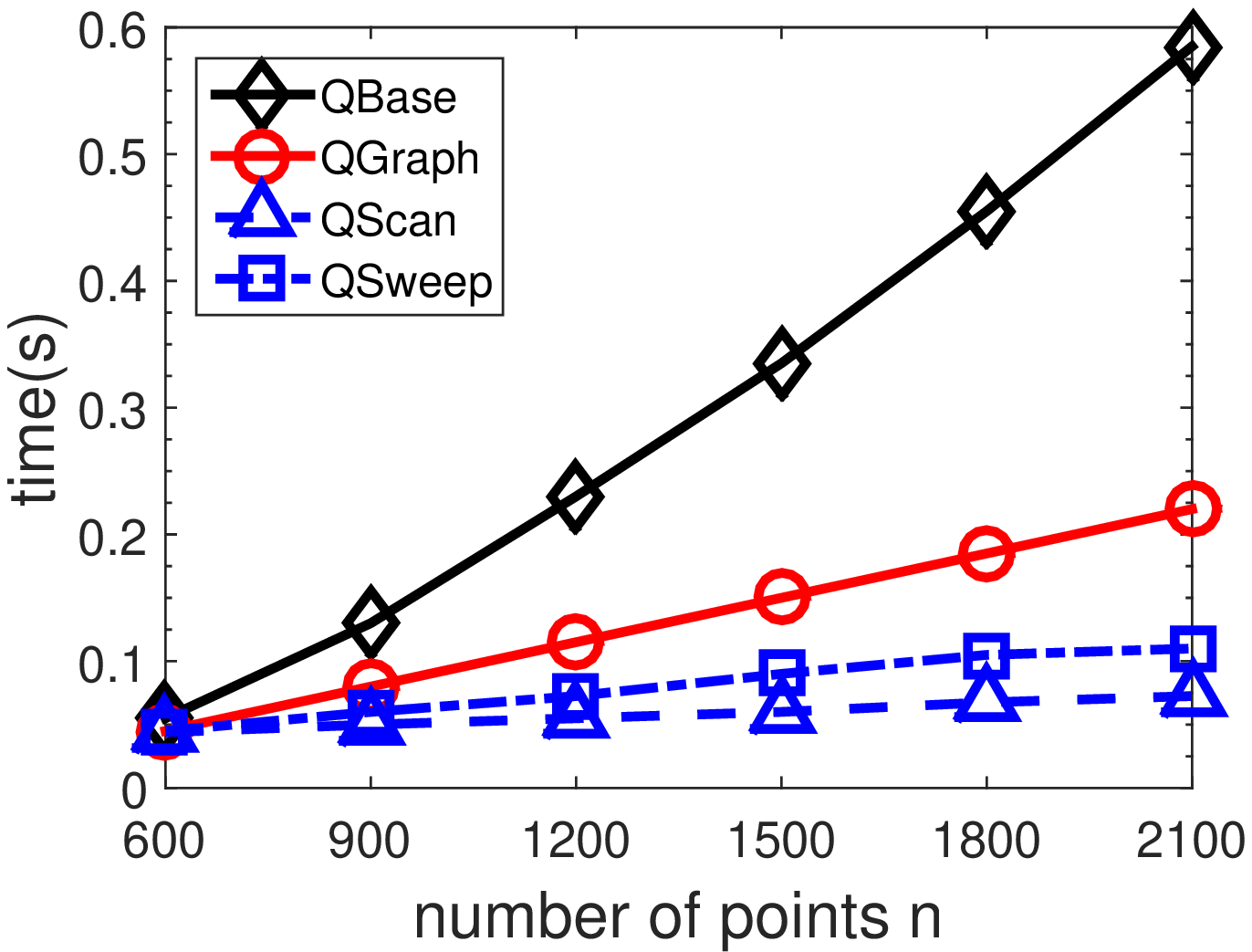}
\end{minipage}
}
 \vspace{-1em}\captionsetup{font={scriptsize}}\caption{The impact of $n$ on skyline diagram of quadrant skyline queries (unlimited domain).} \label{fig:quardant}
\end{figure*}


\section{Experiments}\label{sec:experiments}

In this section, we present experimental studies evaluating our proposed algorithms.

\subsection{Experiment Setup}
We first evaluate the algorithms for computing skyline diagram of quadrant/global skyline, and then the algorithms for dynamic skyline. Finally, we evaluate the algorithms for computing the approximate skyline diagram. We implemented all algorithms in Python and to avoid the effect of I/O, final results are not stored in the exact skyline diagram experiments. We ran experiments on 1) a desktop with Intel Core i7 running Ubuntu 14.04 with 64GB RAM for serial implementations, 2) a computation server with dual Intel Xeon E5-2660 v3 with 512GB RAM running Ubuntu 14.04 for parallel implementations, and 3) a computation server with quad Intel Xeon E5-4627 v3 with 1TB RAM running Ubuntu 16.04 for the approximate skyline diagram. We compare four algorithms (QBase: Baseline algorithm, QGraph: Skyline graph algorithm, QScan: Scanning algorithm, and QSweep: Sweeping algorithm) for quadrant skyline diagram and three algorithms (DBase: Baseline algorithm, DSubset: Subset algorithm, and DScan: Scanning algorithm) for dynamic skyline diagram. We compare two heuristic algorithms (BUM: Bottom-Up Merging algorithm and TDP: Top-Down Partitioning algorithm) for the approximate skyline diagram.

We used both synthetic datasets and a real NBA dataset in our experiments. To study the scalability of our methods, we generated independent (INDE), correlated (CORR), and anti-correlated (ANTI) datasets following the seminal work \cite{DBLP:conf/icde/BorzsonyiKS01}.
We also built a dataset\footnote{Extracted from http://stats.nba.com/leaders/alltime/?ls=iref:nba:gnav on 04/15/2015.} that contains 2384 NBA players who are league leaders of playoffs. Each player has five attributes (Points, Rebounds, Assists, Steals, and Blocks) that measure the player's performance.

\subsection{Skyline Diagram of Quadrant Skyline}

Figures \ref{fig:quardant}(a)(b)(c) present the time cost of QBase, QGraph, QScan, and QSweep with varying number of points $n$ for the three synthetic datasets. For this set of experiments, we used unlimited domains and enforced no two data points lie on the same $x$-coordinate or $y$-coordinate, which can be considered as a stress test for the algorithms. We evaluate the impact of domain size in Section \ref{sec:domainsize}. The results of QBase algorithm on CORR, INDE, and ANTI dataset are almost the same which means the data distribution has no impact on baseline algorithm. We did not report the result of the baseline algorithm in some figures due to the high cost when $n$ is large. All the proposed algorithms scale well with the increasing number of points.

We first examine each algorithm and compare its performance on different datasets. For the QGraph algorithm, the time on INDE dataset is higher than CORR and ANTI datasets.  This is because the number of links between parent and children nodes in the directed skyline graph is larger for INDE dataset. For the QScan algorithm, the time on ANTI dataset is much higher than INDE dataset which is much higher than CORR dataset. This is because the number of skyline in each cell in ANTI dataset is much more than INDE and CORR datasets. Therefore, it requires more time to do the multiset operation on ANTI dataset. For the QSweep algorithm, it is much faster than QGraph and QScan on CORR dataset because there are much fewer intersections thus fewer polyominos on CORR dataset. However, the performance of QSweep is not so good on ANTI dataset due to the huge number of intersections on ANTI dataset.

Comparing different algorithms, QGraph, QScan, and QSweep significantly outperform QBase, which validates the effectiveness of our algorithms. QSweep outperforms QScan on all datasets thanks to its combined steps of finding skyline polyominos directly (but we will see an opposite result on real NBA dataset later). For CORR and INDE datasets, QSweep is the most efficient out of all algorithms, while for ANTI dataset, QGraph has the best performance due to the reason we explained earlier.

Figure \ref{fig:quardant}(d) reports the time cost of QBase, QGraph, QScan, and QSweep with varying number of points $n$ for the real NBA dataset. The difference between the previous synthetic datasets and this NBA dataset is that the latter has a limited domain which leads to fewer number of cells even given the same number of points. Herein, the time cost of $2100$ points on NBA is significantly smaller than that of $2000$ points on synthetic datasets. Comparing different algorithms, the performances of QScan and QSweep are similar and QScan is slightly better than QSweep which is opposite to the performances on synthetic datasets. The reason is that on NBA dataset, the number of cells is much smaller but the number of intersections is similar. However, both QScan and QSweep outperform QGraph.

\subsection{Extension to High-dimensional Space}


Figure \ref{fig:high} reports the time cost of QBase, QGraph, and QScan with varying number of dimensions $d$ for the real NBA dataset. In two-dimensional space, QScan is much better than QGraph, but in high-dimensional space, QScan and QGraph are very similar. The reason is that QScan algorithm needs too many multiset operations in high-dimensional space. Both QGraph and QScan significantly outperform QBase, which verifies the effectiveness and scalability of our proposed algorithms in high-dimensional space.

\begin{figure}[thb]
\begin{minipage}[t]{0.5\linewidth}
\centering

\includegraphics[width=1.6in]{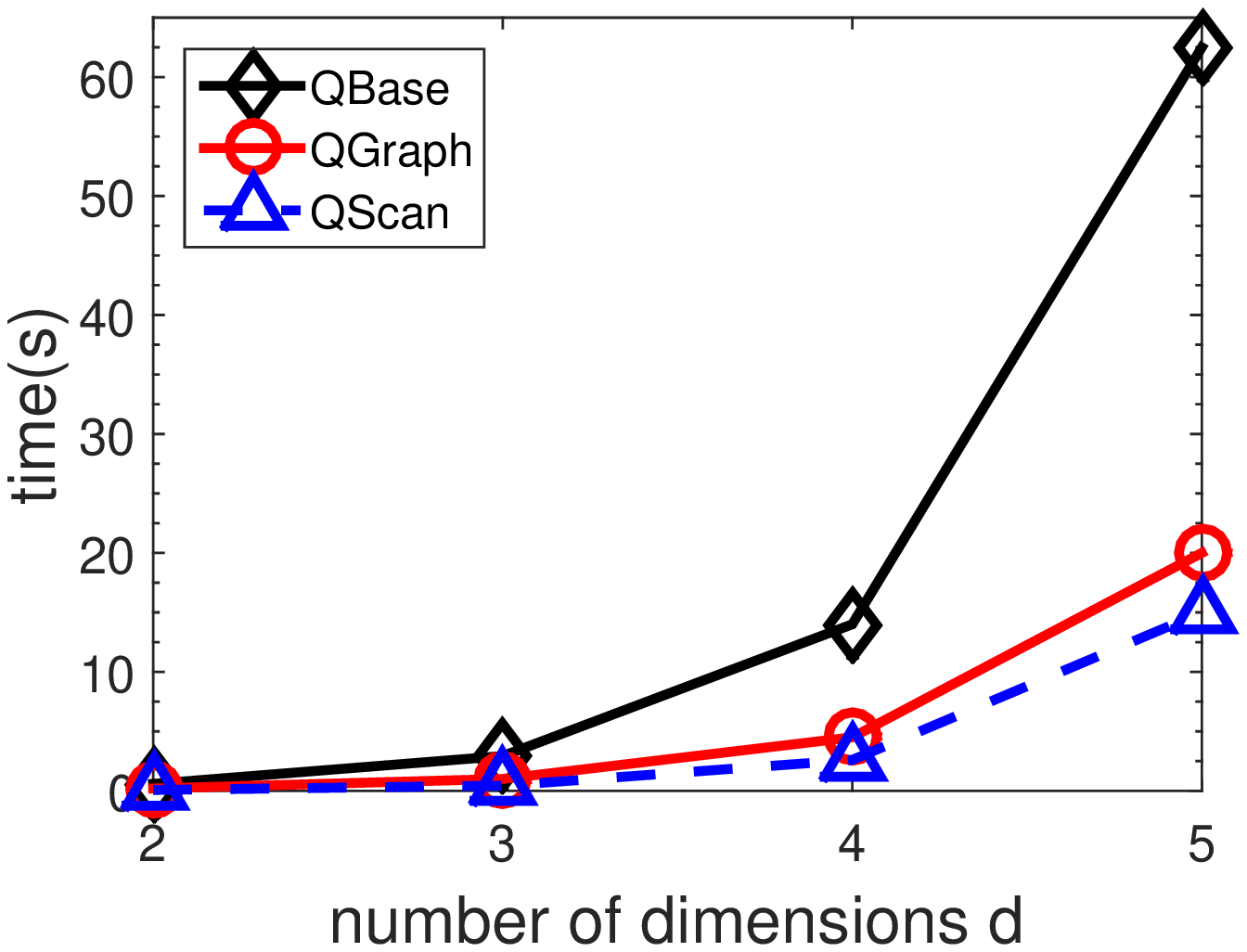}
\captionsetup{font={scriptsize}}
\vspace{-0.5em}
\caption{Impact of dimensions d.} \label{fig:high}
\end{minipage}%
\begin{minipage}[t]{0.5\linewidth}
\centering
\includegraphics[width=1.6in]{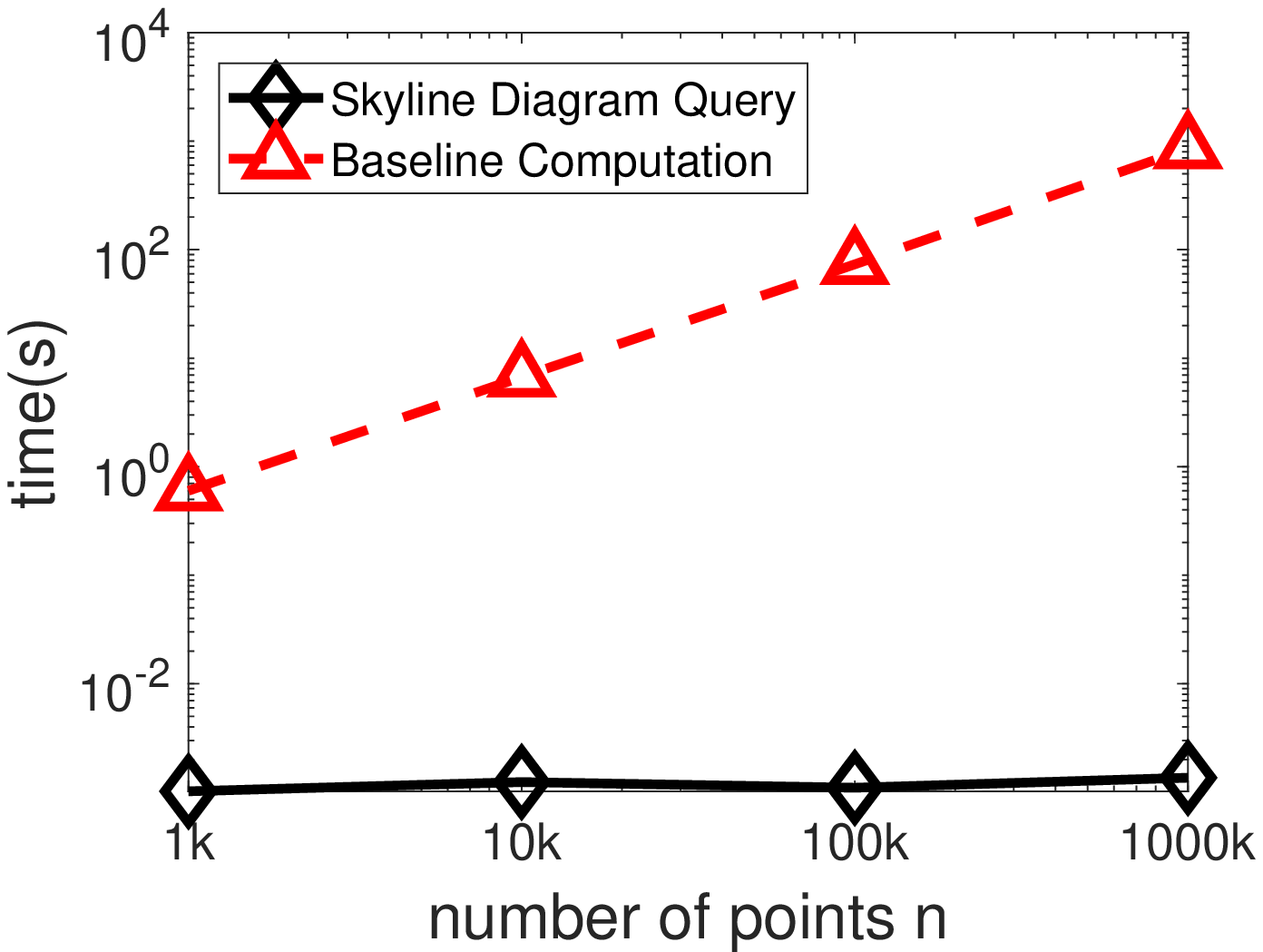}
\captionsetup{font={scriptsize}}
\vspace{-0.5em}
\caption{Query time using skyline diagram.} \label{fig:queryTime}
\end{minipage}
\end{figure}

\subsection{Query Time Using Skyline Diagram}

As we discussed, skyline diagram can be used as a structure for answering skyline queries, reverse skyline queries, as well as other applications. State-of-the-art skyline algorithms without any precomputed structure requires $O(n\log n)$ time ($O(n\log ^{d-1}n)$ for $d$-dimensional space).  Once we have the skyline diagram precomputed, the online time for answering skyline queries can be implemented with only $O(1)$, which is desirable in many real time scenarios. To demonstrate the benefit, we compare the query time using skyline diagram with a skyline query algorithm without precomputed structure. Figure \ref{fig:queryTime} shows the comparison on INDE dataset in two dimensional space. We chose the query point randomly and ran the experiment $1000$ times, the time was accumulated. We can see that the queries based on skyline diagram are $10^5$ times faster and not affected by the increasing number of points, while skyline queries without any structure requires more than one second when $n$ is large.

\subsection{Skyline Diagram of Dynamic Skyline}

\begin{figure*}[]
\centering
\subfigure[\scriptsize{time cost of CORR}]{
\begin{minipage}[b]{0.22\textwidth}
\includegraphics[width=1.1\linewidth]{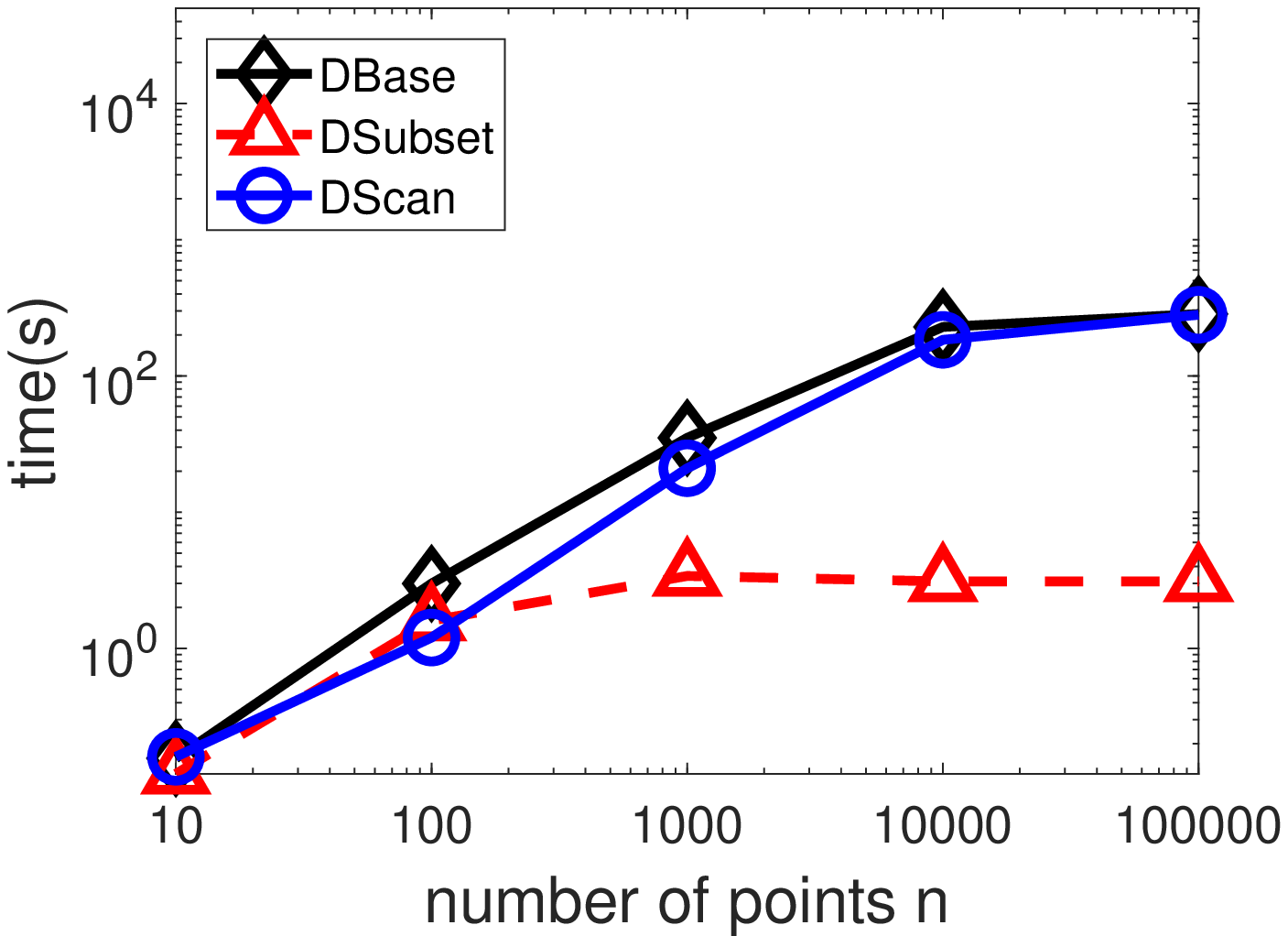}
\end{minipage}
}
\subfigure[\scriptsize{time cost of INDE}]{
\begin{minipage}[b]{0.22\textwidth}
\includegraphics[width=1.1\linewidth]{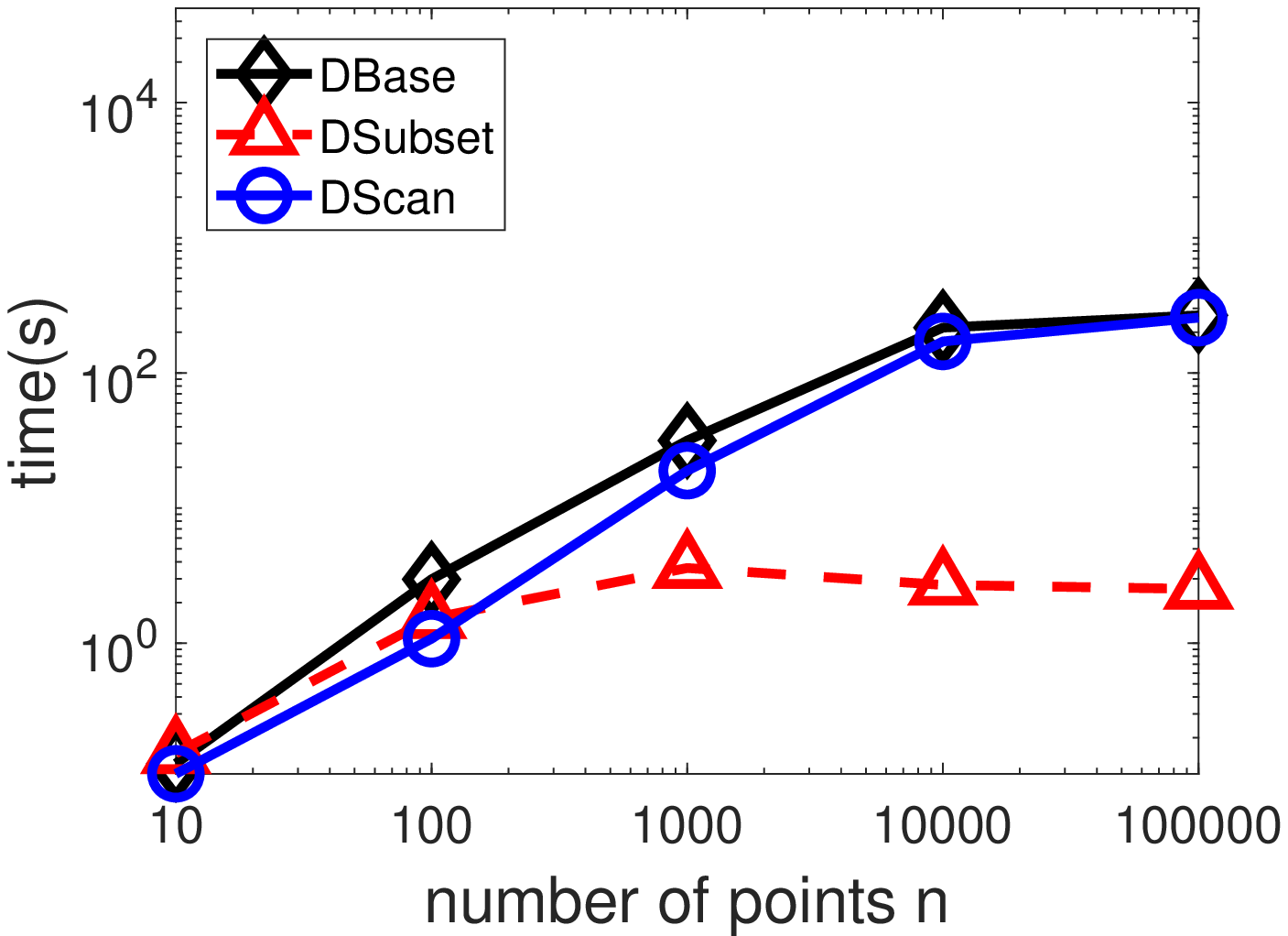}
\end{minipage}
}
\subfigure[\scriptsize{time cost of ANTI}]{
\begin{minipage}[b]{0.22\textwidth}
\includegraphics[width=1.1\linewidth]{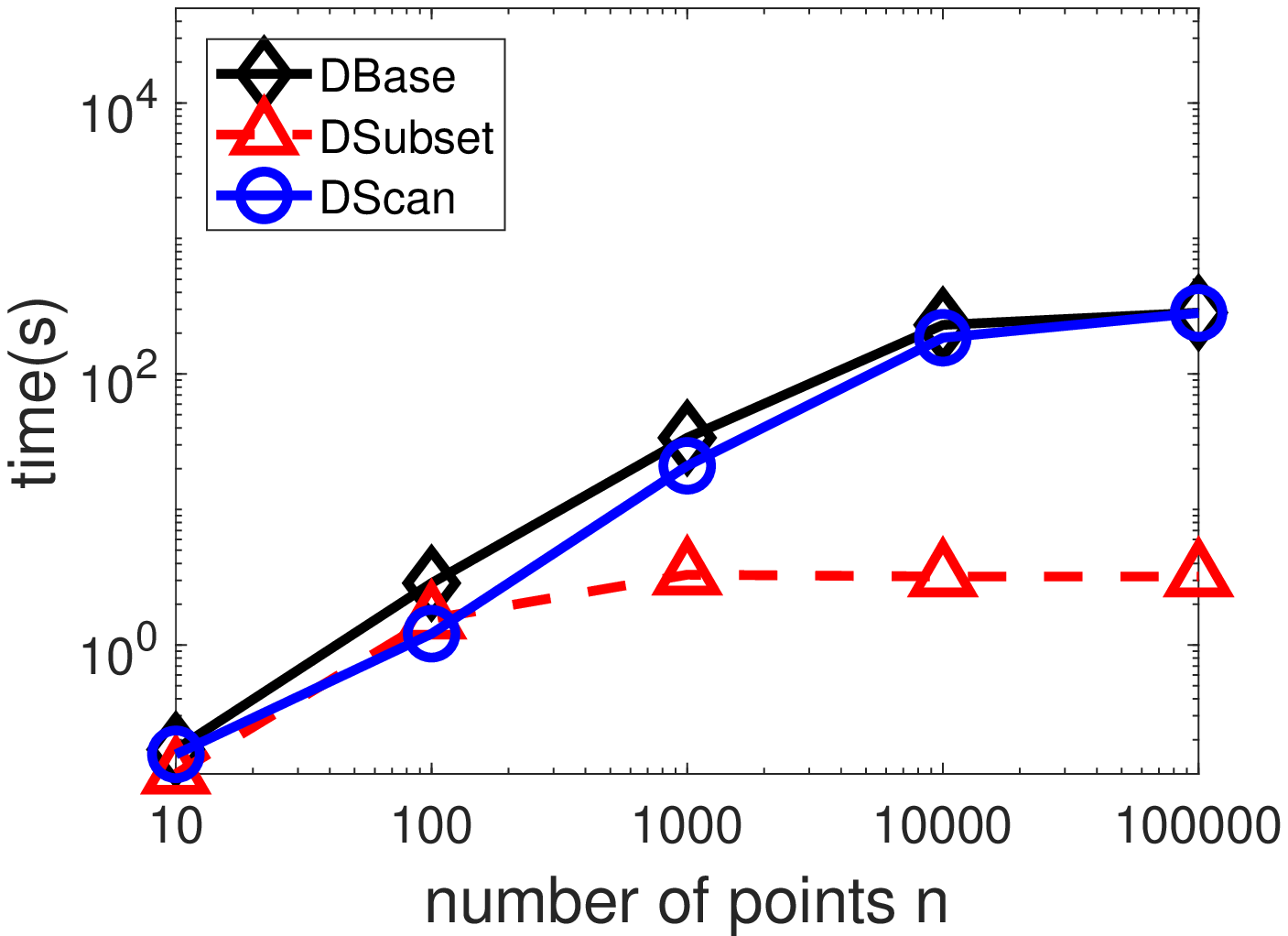}
\end{minipage}
}
\subfigure[\scriptsize{time cost of NBA}]{
\begin{minipage}[b]{0.22\textwidth}
\includegraphics[width=1.1\linewidth]{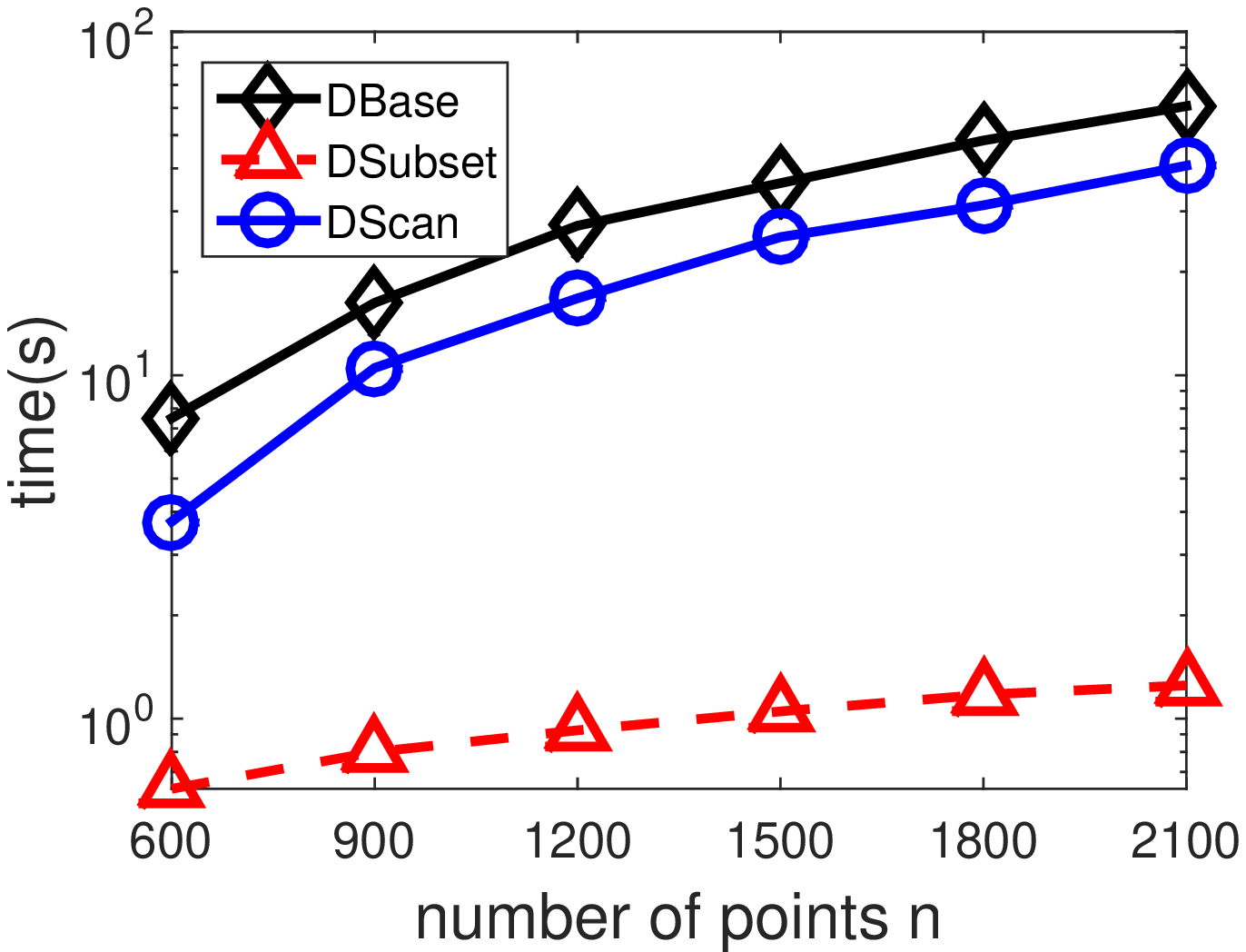}
\end{minipage}
}
 \vspace{-1em}\captionsetup{font={scriptsize}}\caption{The impact of $n$ on skyline diagram of dynamic skyline ($s=10^2$).} \label{fig:dynamic}
\end{figure*}

Figures \ref{fig:dynamic}(a)(b)(c)(d) present the time cost of DBase, DSubset, and DScan with varying number of points $n$ for the three synthetic datasets ($s=10^2$) and the NBA dataset. We used a fixed domain size ($s=10^2$) in this experiment and show the impact of domain size in Section \ref{sec:domainsize}. All algorithms have the same performance on CORR and ANTI datasets. DSubset significantly outperforms DBase, which verifies its effectiveness. For the dataset with large $n$, DSubset significantly outperforms DBase and DScan because when $n$ is much larger than $s$, the number of global skyline is very small, of which dynamic skyline is a subset.

\begin{figure}[H]
\centering
\subfigure[\scriptsize{quadrant skyline queries}]{
\begin{minipage}[b]{0.2\textwidth}
\includegraphics[width=1.1\textwidth]{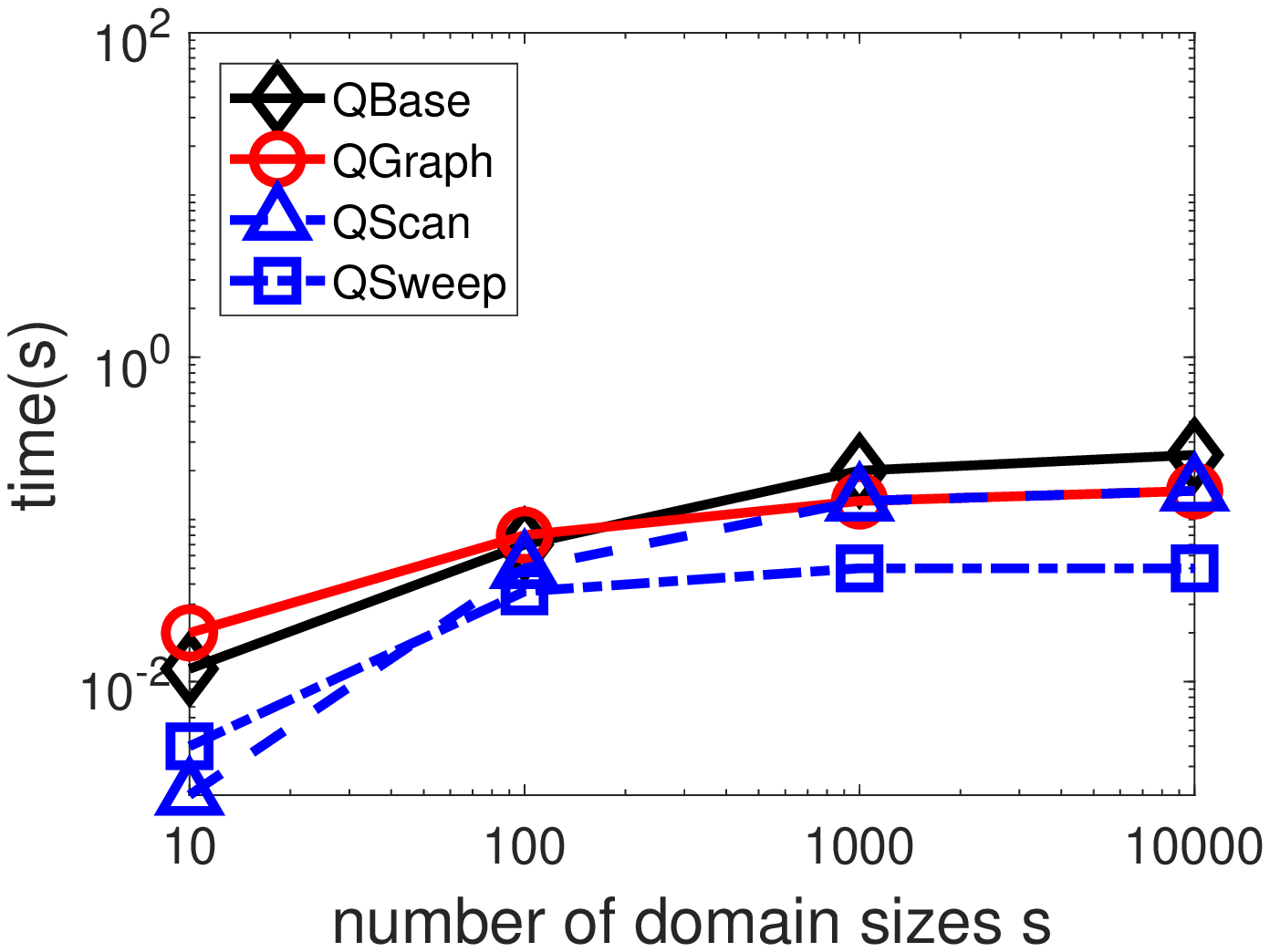}
\end{minipage}
}
\subfigure[\scriptsize{dynamic skyline queries}]{
\begin{minipage}[b]{0.2\textwidth}
\includegraphics[width=1.1\textwidth]{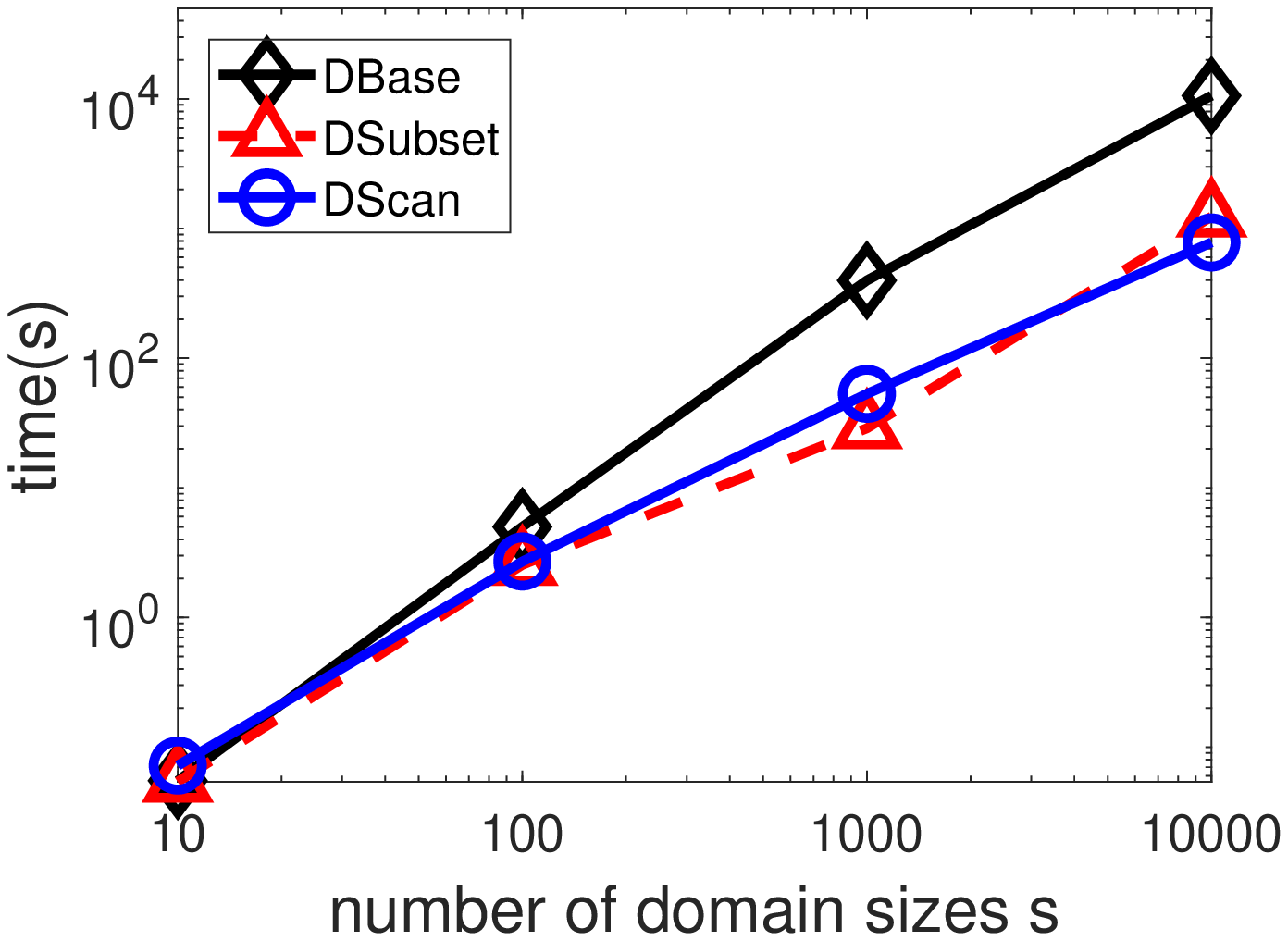}
\end{minipage}
}
 \vspace{-1em}\captionsetup{font={scriptsize}}\caption{The impact of $s$.} \label{fig:diffDomain}
\end{figure}

\vspace{-2em}

\subsection{Impact of Domain Size}\label{sec:domainsize}

In this experiment, we evaluate the impact of domain size on both quadrant and dynamic skyline diagram algorithms.  Figure \ref{fig:diffDomain} reports the time cost of different algorithms with varying domain size $s$ on INDE dataset ($n=200, d=2$). We observe that the time increases with increasing $s$ as expected.  On the other hand, when $s$ is much larger than $n$, increasing $s$ does not have an impact unless $n$ increases.  In addition, when $s$ is much larger than $n$, we see that DScan outperforms DSubset because the number of global skyline is very large in dataset with large domains.

\subsection{Approximate Skyline Diagram}
In this subsection, we evaluate two heuristic algorithms for the approximate skyline diagram in terms of the time cost, the space cost, and the precision. We define the precision as the average ratio of the number of skyline points in each skyline cell and the number of skyline points in each skyline polyomino that contains this skyline cell, $\sum_{i=1,2,...,n;j=1,2,...,n}\frac{|Sky(C_{i,j})|}{|Sky(SP_k)|}$, where skyline polyomino $SP_k$ contains skyline cell $C_{i,j}$. We note that the precision of an exact skyline diagram we have evaluated so far is $100\%$ since all skyline cells within a skyline polyomino are guaranteed to have the same skyline result.

\begin{figure}[H]
\centering
\subfigure[\scriptsize{time cost}]{
\begin{minipage}[b]{0.149\textwidth}
\includegraphics[width=1.1\textwidth]{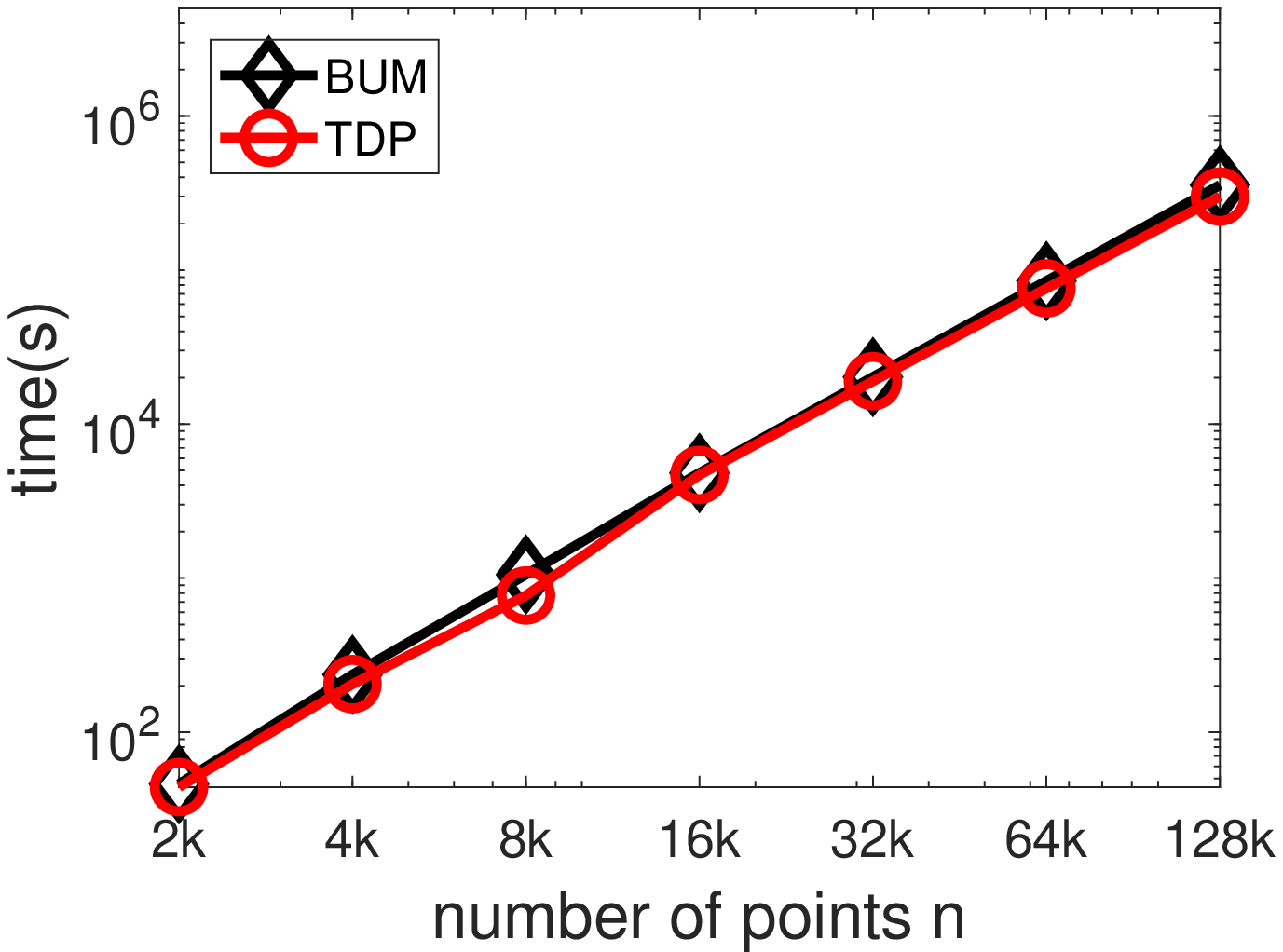}
\end{minipage}
}
\subfigure[\scriptsize{space cost}]{
\begin{minipage}[b]{0.149\textwidth}
\includegraphics[width=1.1\textwidth]{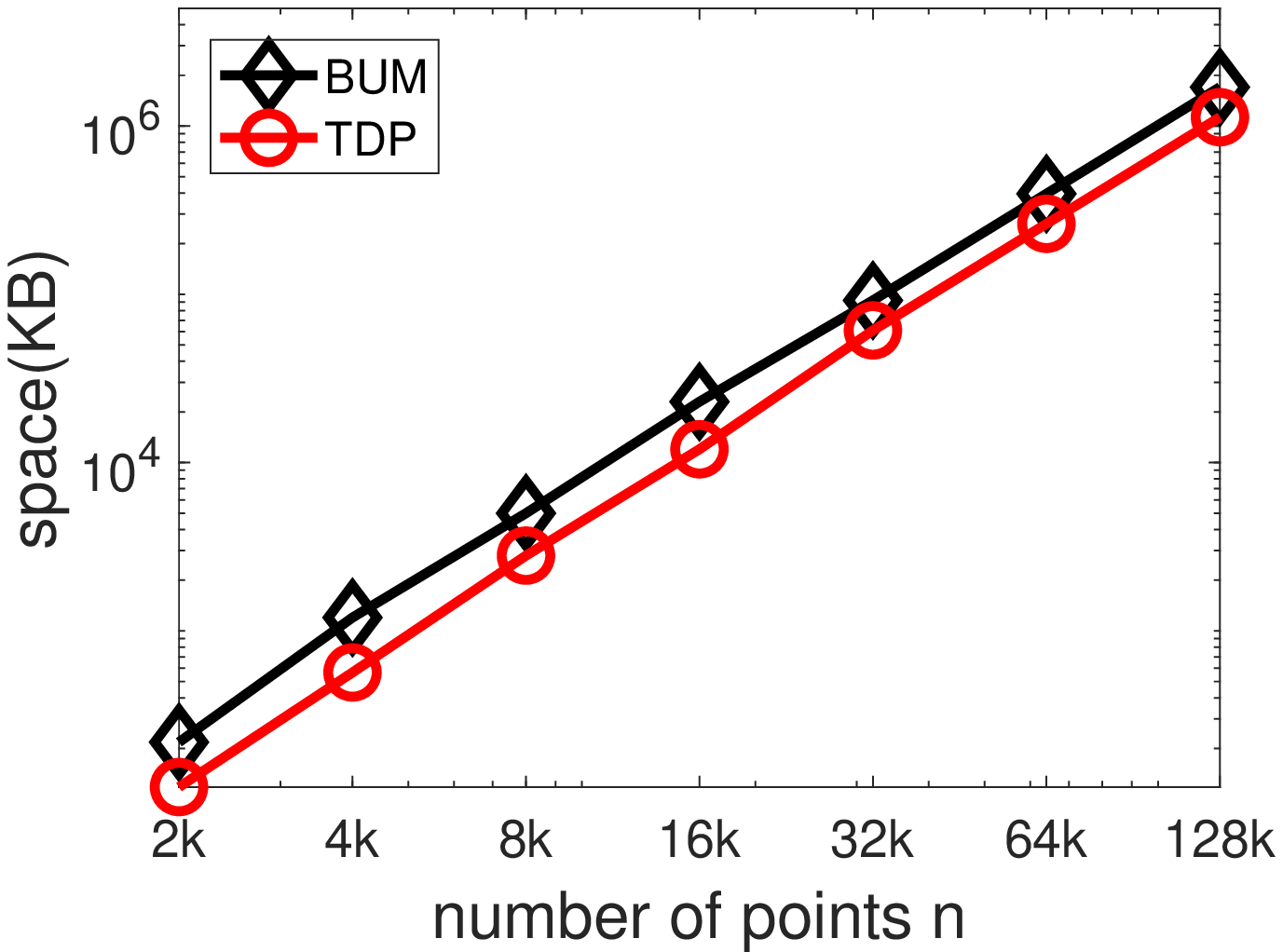}
\end{minipage}
}
\subfigure[\scriptsize{precision}]{
\begin{minipage}[b]{0.149\textwidth}
\includegraphics[width=1.1\textwidth]{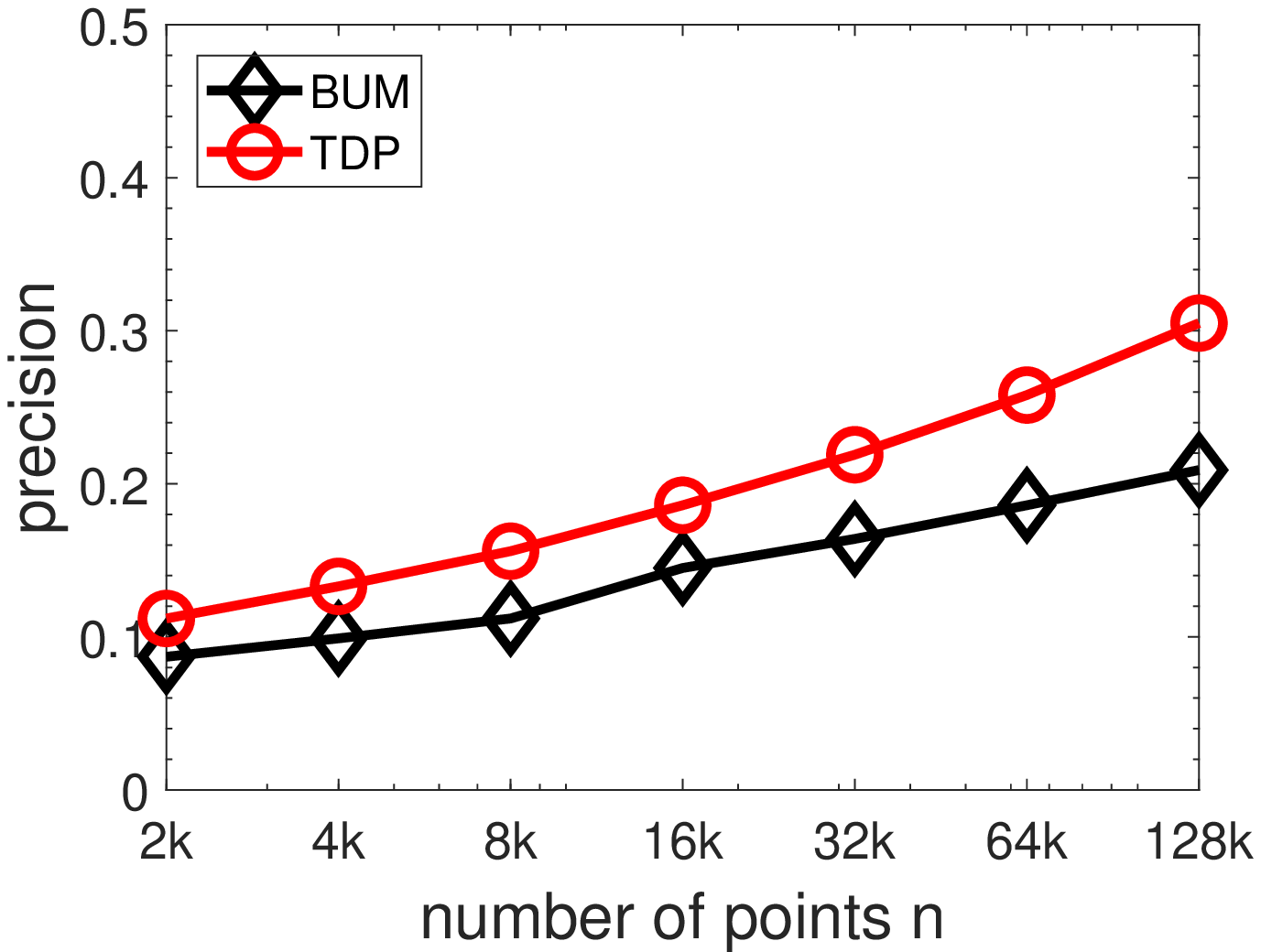}
\end{minipage}
}
 \vspace{-1em}\captionsetup{font={scriptsize}}\caption{The impact of $n$ ($\delta$=90).} \label{fig:AdiffN}
\end{figure}

\begin{figure}[H]
\centering
\subfigure[\scriptsize{time cost}]{
\begin{minipage}[b]{0.149\textwidth}
\includegraphics[width=1.1\textwidth]{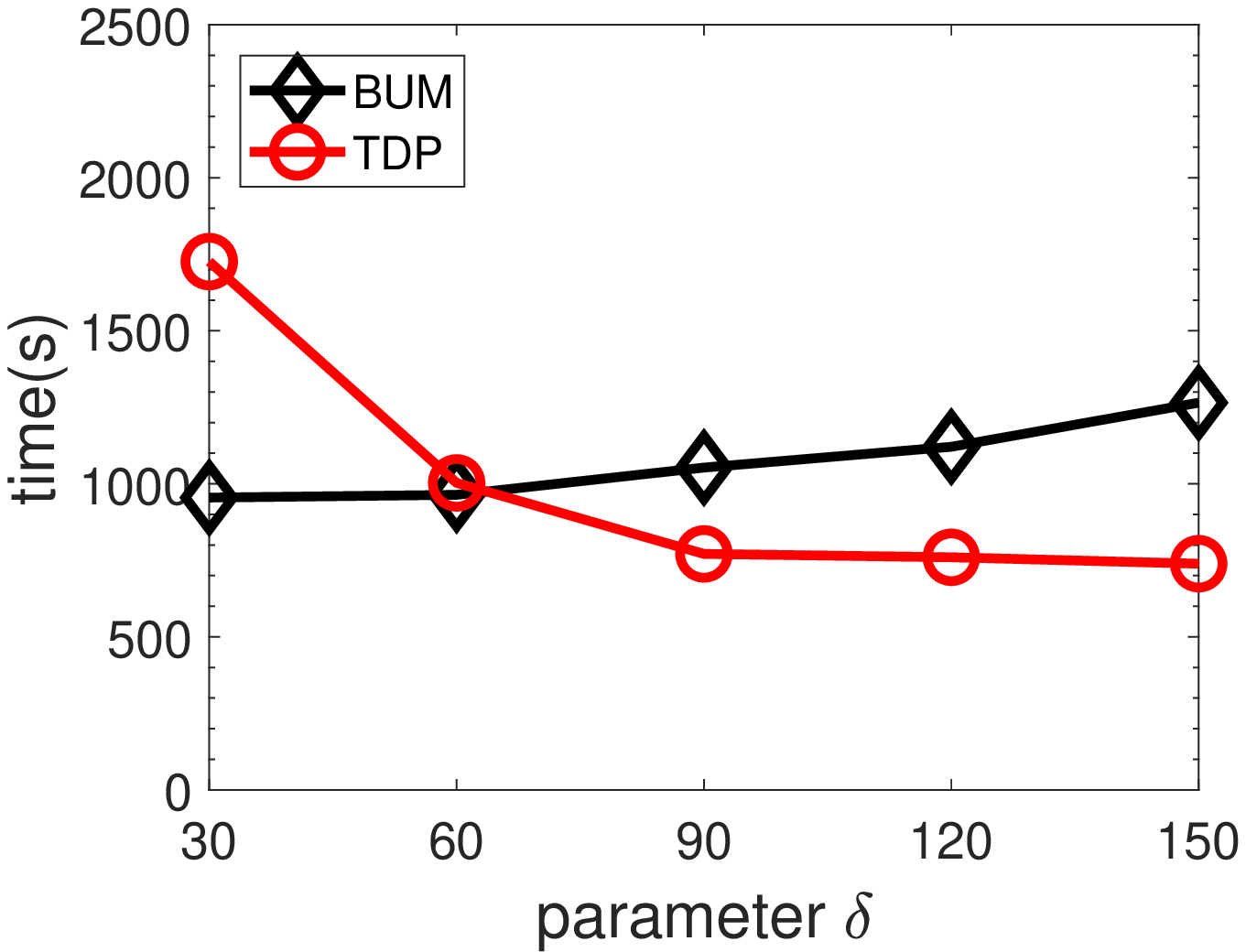}
\end{minipage}
}
\subfigure[\scriptsize{space cost}]{
\begin{minipage}[b]{0.149\textwidth}
\includegraphics[width=1.1\textwidth]{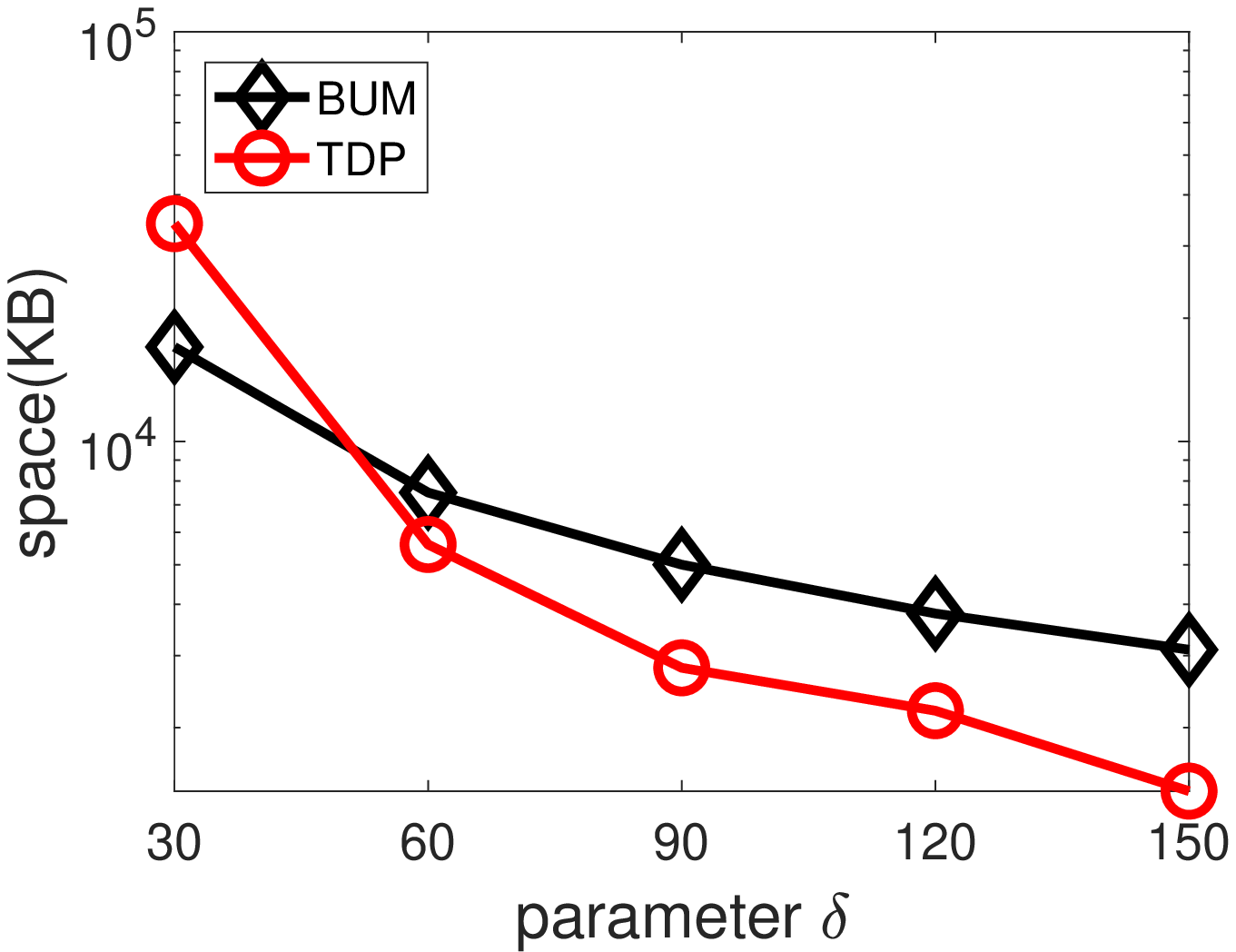}
\end{minipage}
}
\subfigure[\scriptsize{precision}]{
\begin{minipage}[b]{0.149\textwidth}
\includegraphics[width=1.1\textwidth]{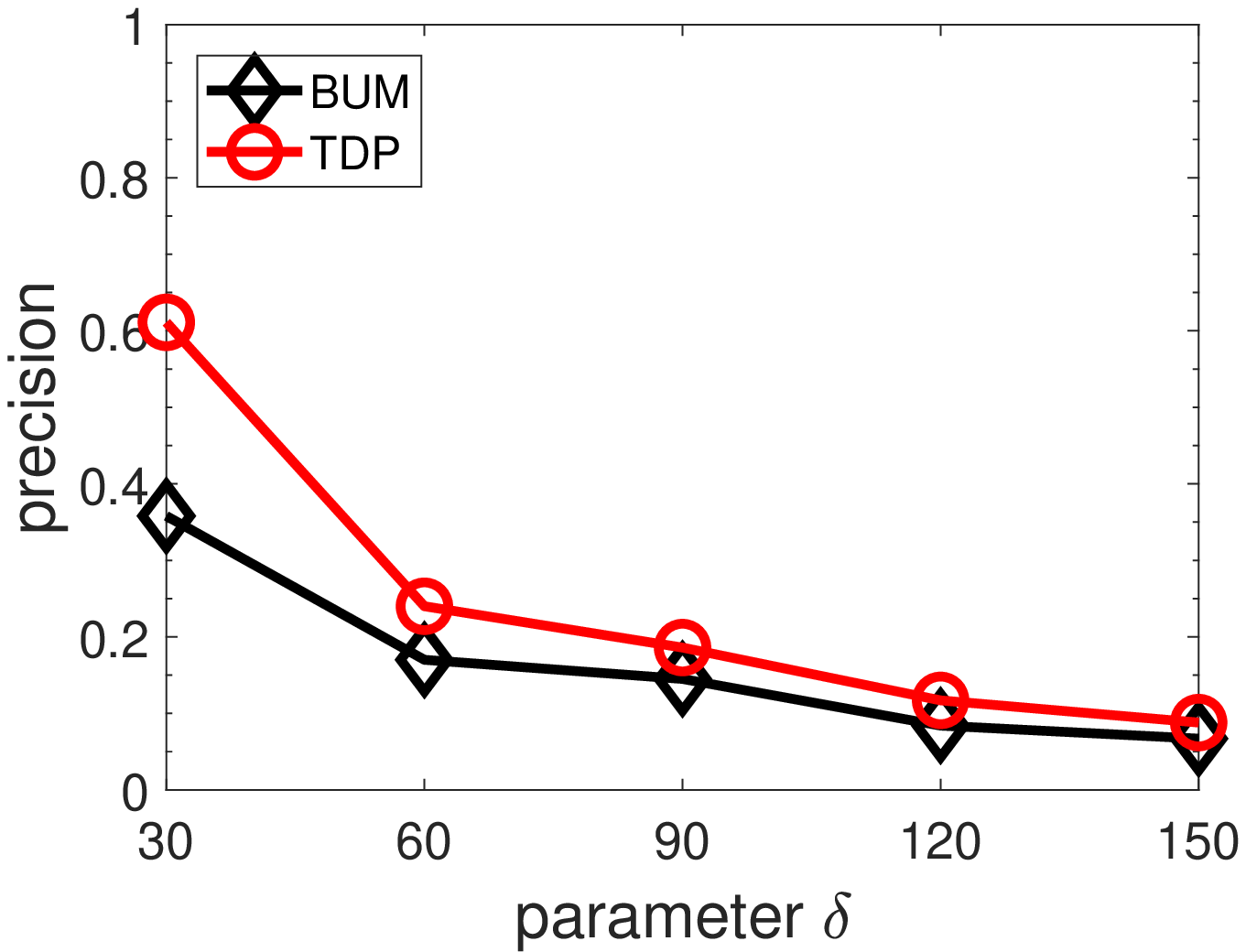}
\end{minipage}
}
 \vspace{-1em}\captionsetup{font={scriptsize}}\caption{The impact of $\delta$ (n=8k).} \label{fig:AdiffDelta}
\end{figure}

Figures \ref{fig:AdiffN}(a)(b)(c) present the impact of the number of points $n$ on the time cost, the space cost, and the precision. Both the time cost and the space cost increase linearly with the increasing number of points $n$. Precision increases with the increasing number of points $n$ because the number of skyline points in each skyline cell is increasing and the number of skyline points in each skyline polyomino does not change substantially as we fix $\delta=90$.

Figures \ref{fig:AdiffDelta}(a)(b)(c) present the impact of parameter $\delta$ on the time cost, the space cost, and the precision. In Figure \ref{fig:AdiffDelta}(a), BUM is better than TDP when $\delta$ is small, but is worse when $\delta$ is large. The time cost of BUM increases with increasing $\delta$ because we need to check more skyline cells to find HPLs and VPLs. On the contrary, the time cost of TDP decreases with increasing $\delta$ because the needed number of guessing the exact number of partitioning lines decreases. Therefore, we can employ BUM and TDP when $\delta$ is small and large, respectively. Furthermore, if we can learn the appropriate number of partitioning lines rather than guessing from 2, TDP will have a much better performance. In Figure \ref{fig:AdiffDelta}(b), the space cost for both BUM and TDP decreases with increasing $\delta$ because we need less skyline polyominos when $\delta$ is large. When $\delta$ is small, the space cost of TDP is larger than the space cost of BUM, but is smaller when $\delta$ is large, which corresponds to the trend of the time cost in Figure \ref{fig:AdiffDelta}(a). In Figure \ref{fig:AdiffDelta}(c), precision decreases with the increasing $\delta$ because the number of skyline points in each skyline cell does not change substantially but the number of skyline points in each skyline polyomino is increasing.

\section{Conclusions and Future Work}\label{sec:conclusion}
In this paper, we proposed a novel concept called skyline diagram.  Given a set of points, it partitions the plane into a set of skyline polyominos where query points in each polyomino have the same skyline query results. We studied skyline diagram for three kinds of skyline queries and presented several efficient algorithms to compute the skyline diagram. We propose two heuristic algorithms, bottom-up merging algorithm and top-down partitioning algorithm, to efficiently compute the approximate skyline diagram with different tradeoffs. Experimental results on both real and synthetic datasets show that our algorithms are efficient and scalable.

{\tiny
\bibliographystyle{abbrv}\tiny
\bibliography{SkylineCell}
}

\vspace{-1em}
\begin{IEEEbiography}[{\includegraphics[width=1in,height=1.25in,clip,keepaspectratio]{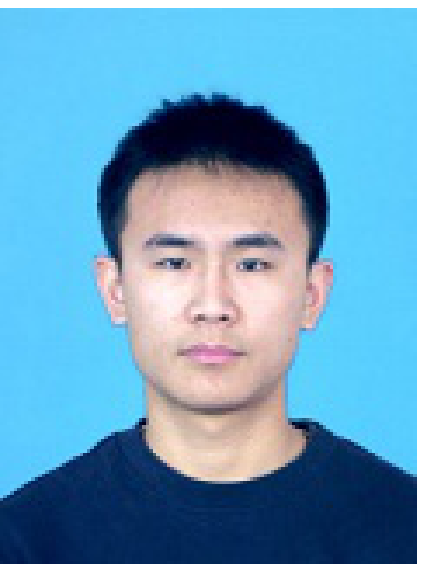}}]{Jinfei Liu}
is a joint postdoctoral research fellow at  Emory University and Georgia Institute of Technology. His research interests include skyline queries, data privacy and security, and machine learning. He has published over 20 papers in premier journals and conferences including TKDE, VLDB, ICDE, CIKM, and IPL.
\end{IEEEbiography}

\vspace{-1em}
\begin{IEEEbiography}[{\includegraphics[width=1in,height=1.25in,clip,keepaspectratio]{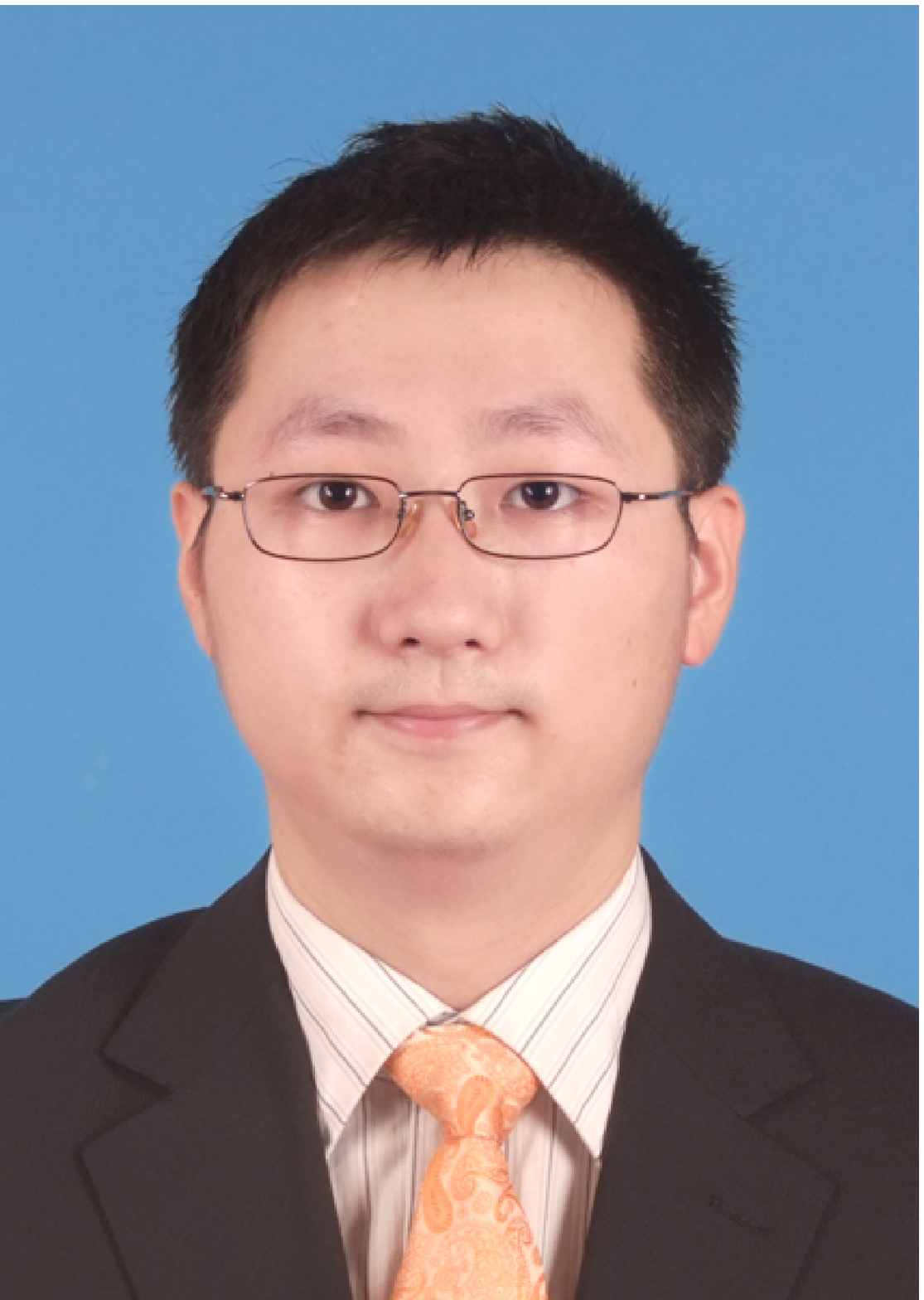}}]{Juncheng Yang}
is Ph.D. student at Carnegie Mellon University. His research interests include computer security, database, smart cache in storage and distributed system. He has published over 10 papers in premier conferences including ICDE and SoCC.
\end{IEEEbiography}

\vspace{-1em}
\begin{IEEEbiography}[{\includegraphics[width=1in,height=1.25in,clip,keepaspectratio]{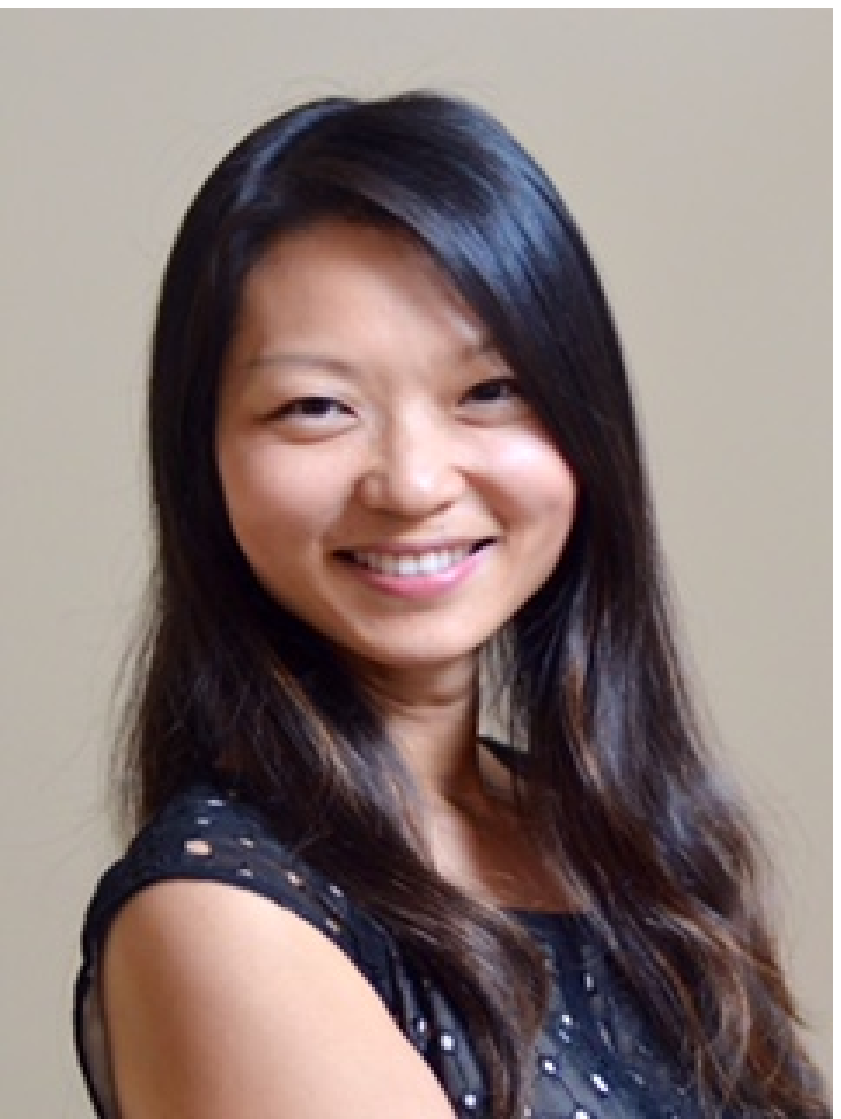}}]{Li Xiong}
is a Professor of Computer Science and Biomedical Informatics at Emory University. She conducts research that addresses both fundamental and applied questions at the interface of data privacy and security, spatiotemporal data management, and health informatics.  She has published over 100 papers in premier journals and conferences including TKDE, JAMIA, VLDB, ICDE, CCS, and WWW. She currently serves as associate editor for IEEE Transactions on Knowledge and Data Engineering (TKDE) and on numerous program committees for data management and data security conferences.
\end{IEEEbiography}

\vspace{-1em}
\begin{IEEEbiography}[{\includegraphics[width=1.5in,height=1.35in,clip,keepaspectratio]{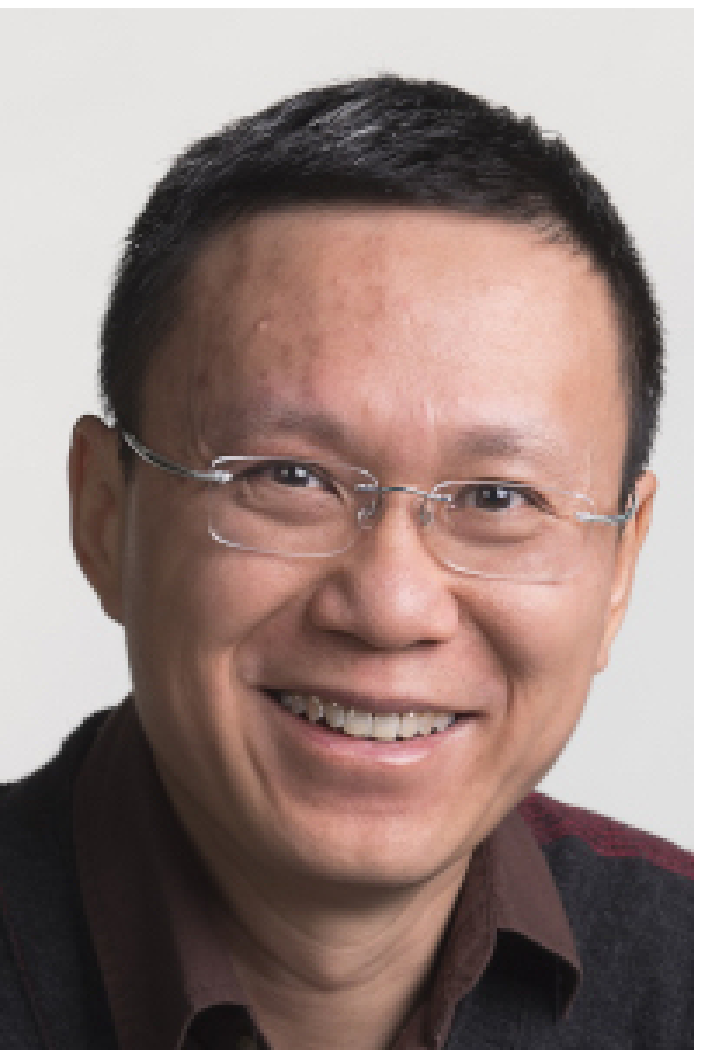}}]{Jian Pei}
is currently a Canada Research Chair (Tier 1) in Big Data Science, a Professor in the School of Computing Science at Simon Fraser University, Canada. He is one of the most cited authors in data mining, database systems, and information retrieval. Since 2000, he has published one textbook, two monographs and over 200 research papers in refereed journals and conferences, which have been cited by more than 77,000 in literature. He was the editor-in-chief of the IEEE Transactions of Knowledge and Data Engineering (TKDE) in 2013-2016, is currently a director of the Special Interest Group on Knowledge Discovery in Data (SIGKDD) of the Association for Computing Machinery (ACM). He is a Fellow of the ACM and of the IEEE.
\end{IEEEbiography}

\vspace{-1em}
\begin{IEEEbiography}[{\includegraphics[width=1in,height=1.35in,clip,keepaspectratio]{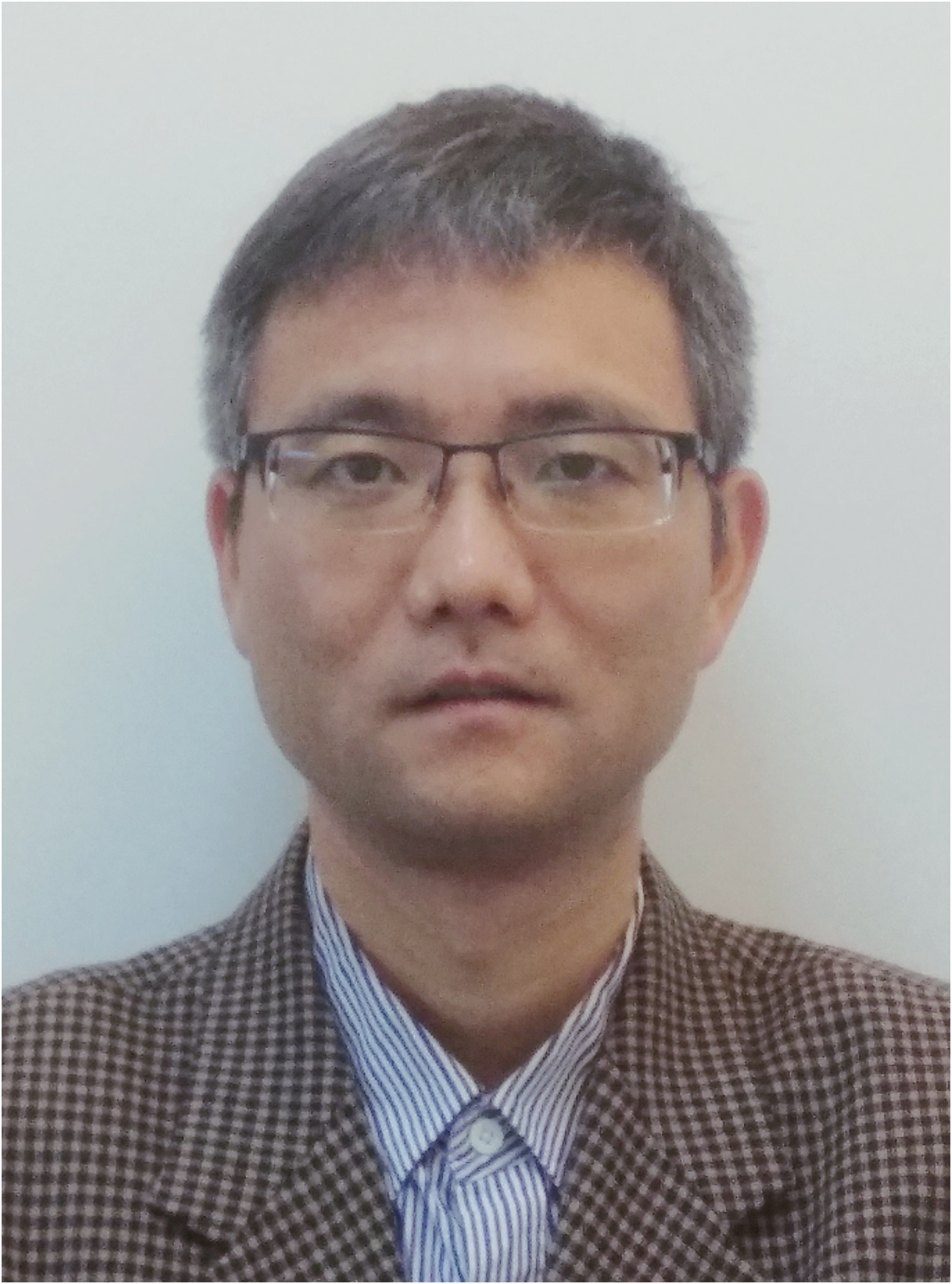}}]{Jun Luo}
is a principal researcher at Lenovo Machine Intelligence Center in Hong Kong. He received his PhD degree in computer science from the University of Texas at Dallas, USA, in 2006. His research interests include big data, machine learning, spatial temporal data mining and computational geometry. He has published over 90 journal and conference papers in these areas.
\end{IEEEbiography}

\vspace{-1em}
\begin{IEEEbiography}[{\includegraphics[width=1in,height=1.35in,clip,keepaspectratio]{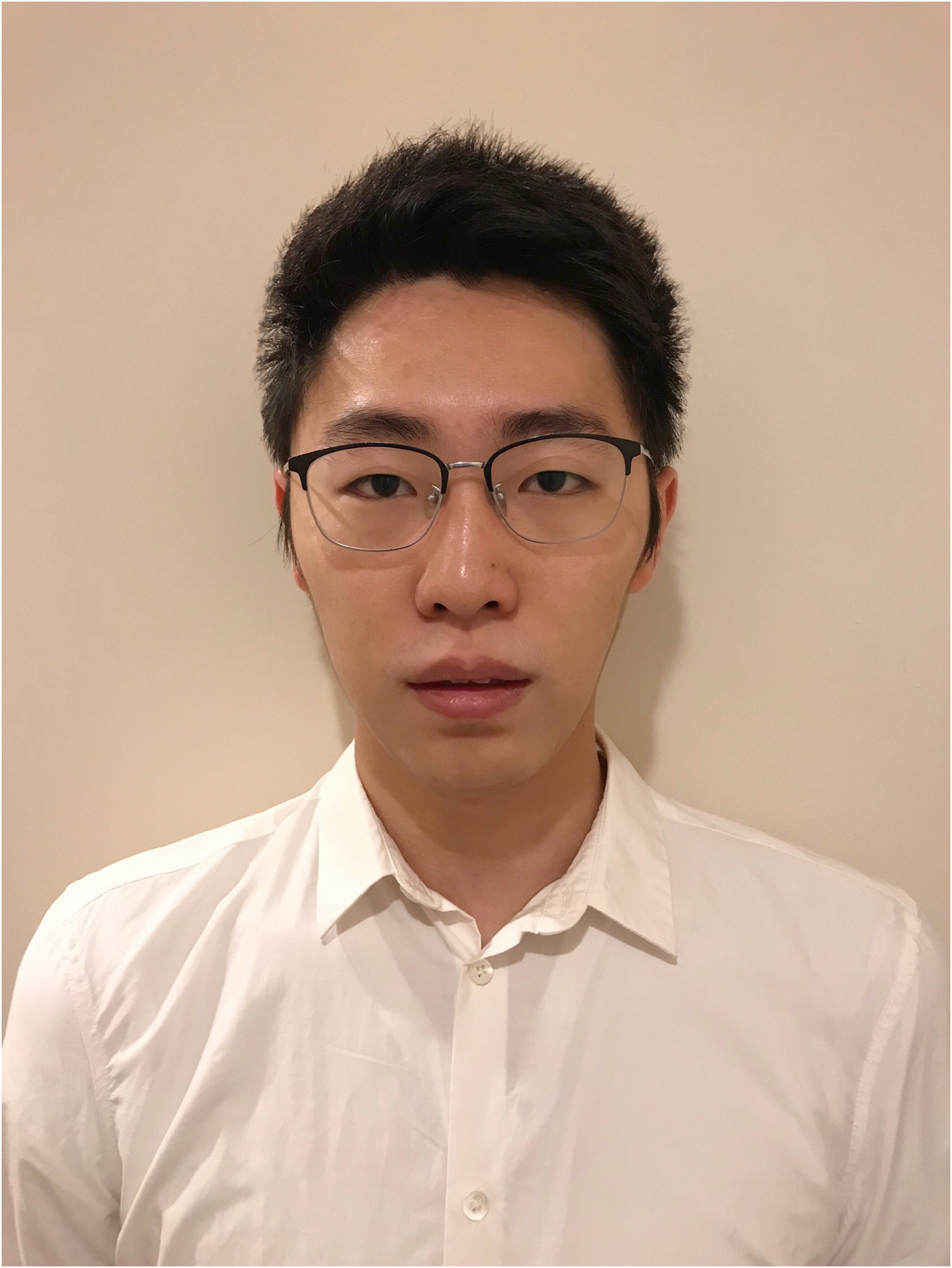}}]{Yuzhang Guo}
is an undergraduate student at Emory University. His research interests include machine learning and data science.
\end{IEEEbiography}

\vspace{-1em}
\begin{IEEEbiography}[{\includegraphics[width=1in,height=1.35in,clip,keepaspectratio]{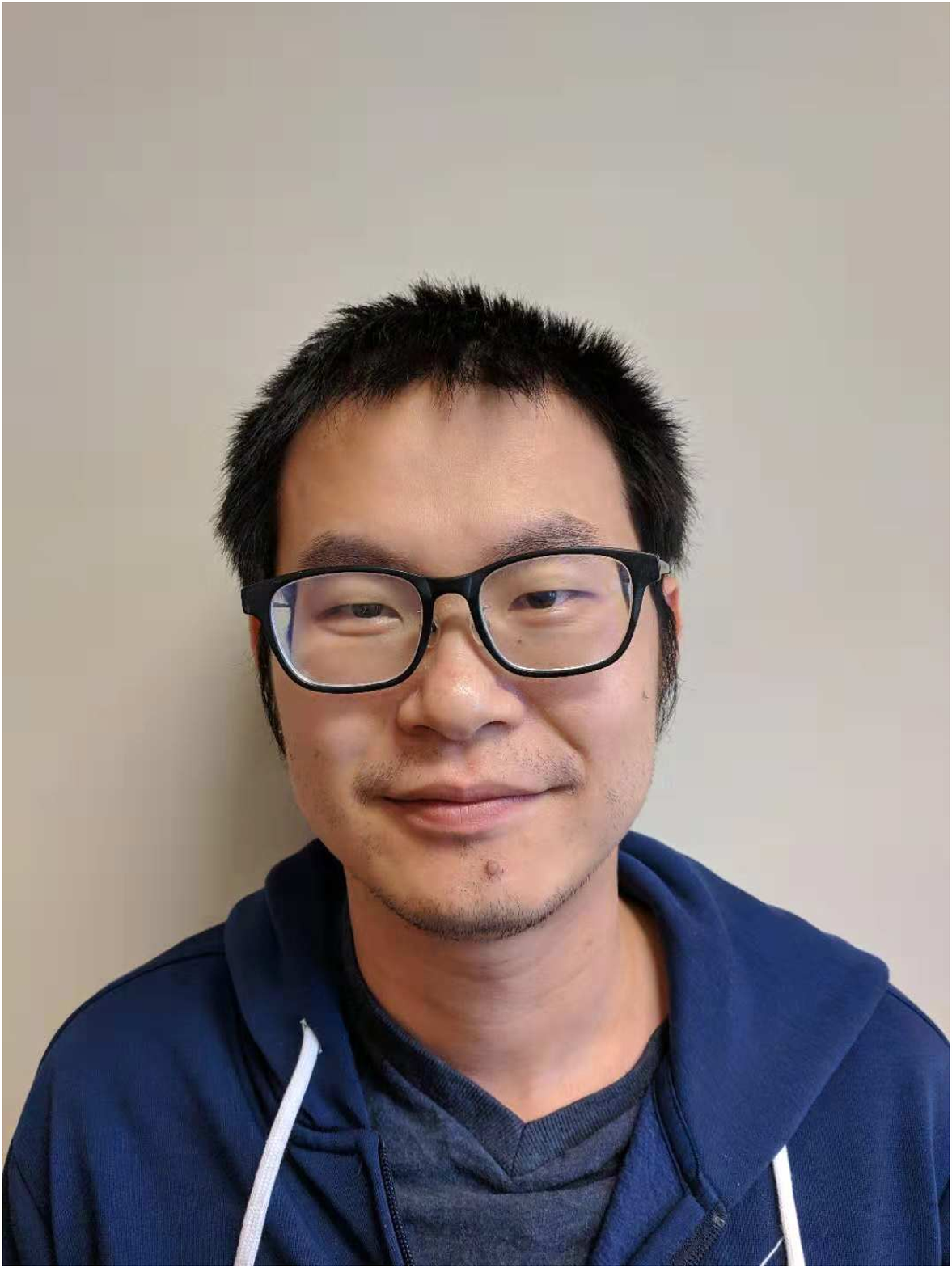}}]{Shuaicheng Ma}
is a master student at University of Central Florida. He is currently a visiting researcher at Emory University. His research interests include data privacy, security, and blockchain.
\end{IEEEbiography}

\vspace{-1em}
\begin{IEEEbiography}[{\includegraphics[width=1in,height=1.35in,clip,keepaspectratio]{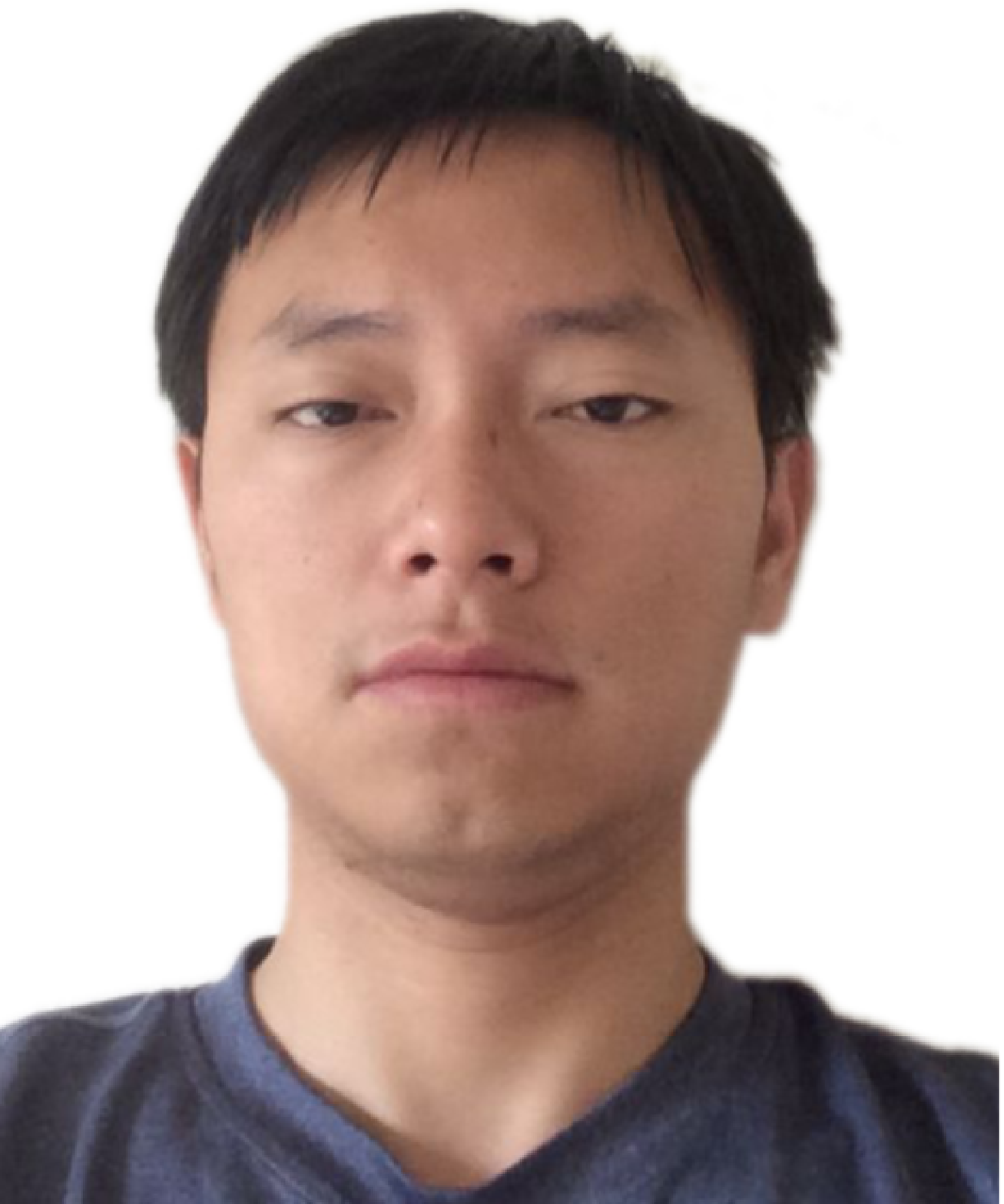}}]{Chenglin Fan}
is a Ph.D. candidate at University of Texas at Dallas. His research interests including algorithm theory, computational geometry and data science. He has published over 10 papers in premier conferences including SODA and SoCG.
\end{IEEEbiography}

\end{document}